\documentclass{article}
\usepackage{amsmath,amssymb,amsthm}
\usepackage{mathtools}
\usepackage{xcolor}
\usepackage{hyperref}
\usepackage{geometry}
\usepackage{listings}
\usepackage{url}
\usepackage{booktabs}
\usepackage{graphicx}
\usepackage{enumitem} 
\usepackage{microtype} 
\usepackage{tikz}      
\usepackage{float}     

\theoremstyle{definition}
\newtheorem{definition}{Definition}[section]
\newtheorem{theorem}[definition]{Theorem}
\newtheorem{lemma}[definition]{Lemma}

\newtheorem{proposition}[definition]{Proposition}
\newtheorem{example}{Example}[section]
\newtheorem{remark}{Remark}[section]

\newcommand{\Lean}[1]{\texttt{#1}}
\newcommand{\R}{\mathbb{R}}
\newcommand{\Icc}[2]{[#1, #2]}
\newcommand{\Ioo}[2]{(#1, #2)}

\newcommand{\Fin}{\text{Fin}}
\newcommand{\RVector}{\text{RVector}}
\newcommand{\Iccn}{\text{Icc\_n}}
\newcommand{\Ioon}{\text{Ioo\_n}}
\newcommand{\indicatorURO}{\text{indicatorUpperRightOrthant}}
\newcommand{\survivalProbN}{\text{survivalProbN}}
\newcommand{\mixedVec}{\text{mixedVector}}
\newcommand{\riemannStieltjesND}{\text{riemannStieltjesIntegralND}}
\newcommand{\allGt}{\text{allGt}}

\newcommand{\tactic}[1]{\textcolor{blue}{\texttt{#1}}}
\newcommand{\lean}[1]{\texttt{#1}}
\newcommand{\indicator}[1]{\mathbf{1}_{#1}}

\hypersetup{
  colorlinks=true,
  linkcolor=blue,
  filecolor=magenta,
  urlcolor=cyan,
  citecolor=green!50!black,
  pdftitle={Geometric Formalization of Stochastic Dominance: A Tractable Path for Multi-Dimensional Economic Analysis}
}

\title{Geometric Formalization of First-Order Stochastic Dominance in $N$ Dimensions: \\ A Tractable Path to Multi-Dimensional Economic Decision Analysis}
\author{Jingyuan Li\thanks{Email: \texttt{jingyuanli@ln.edu.hk}. Department of Operations and Risk Management, Lingnan University.}}
\date{May 17, 2025}

\begin{document}
\maketitle

\begin{abstract}
This paper introduces and formally verifies a novel geometric framework for first-order stochastic dominance (FSD) in $N$ dimensions using the Lean 4 theorem prover. Traditional analytical approaches to multi-dimensional stochastic dominance rely heavily on complex measure theory and multivariate calculus, creating significant barriers to formalization in proof assistants. Our geometric approach characterizes $N$-dimensional FSD through direct comparison of survival probabilities in upper-right orthants, bypassing the need for complex integration theory. We formalize key definitions and prove the equivalence between traditional FSD requirements and our geometric characterization. This approach achieves a more tractable and intuitive path to formal verification while maintaining mathematical rigor. We demonstrate how this framework directly enables formal analysis of multi-dimensional economic problems in portfolio selection, risk management, and welfare analysis. The work establishes a foundation for further development of verified decision-making tools in economics and finance, particularly for high-stakes domains requiring rigorous guarantees.
\end{abstract}

\begin{flushleft}
\textbf{Keywords:} Geometric Stochastic Dominance, Formal Verification, Lean 4, Interactive Theorem Proving, Multi-Dimensional Decision Theory, Portfolio Selection, Risk Management, Welfare Analysis,  Certified Economic Modeling

\end{flushleft}
\smallskip
\begin{flushleft}
\textbf{JEL Classification:} D81 (Decision-Making under Risk and Uncertainty), C65 (Miscellaneous Mathematical Tools), C63 (Computational Techniques)
\end{flushleft}

\newpage

\section{Introduction}\label{sec:Intro}
Decision-making under uncertainty is a cornerstone of economic theory. Stochastic dominance (SD) offers a robust and widely accepted framework for comparing risky prospects without requiring precise specification of utility functions. When one prospect stochastically dominates another, it provides unambiguous guidance for rational decision-making, making it a powerful tool across economics, finance, and welfare analysis\cite{HadarRussell1969}\cite{HanochLevy1969}\cite{LevyParouch1974}\cite{Scarsini1988}\cite{RussellSeo1989}\cite{Gollier2001}\cite{MullerStoyan2002}
\cite{Denuitetal2005}\cite{ShakedShanthikumar2007}\cite{DenuitMesfioui2010}\cite{Denuitetal2013}\cite{Levy2015}.

While one-dimensional SD is well-understood and relatively straightforward to analyze, extending SD concepts to $N$ dimensions—where outcomes are vectors of attributes—introduces significant mathematical and computational challenges. Traditional approaches rely heavily on measure theory, multivariate calculus, and complex integration techniques, making formalization in proof assistants particularly difficult. This complexity limits the application of formal verification to multi-dimensional economic decision problems, precisely where formal guarantees would be most valuable due to the high stakes and complexity involved.

This paper confronts this challenge by introducing a novel geometric approach to formalize and verify $N$-dimensional first-order stochastic dominance (FSD) within the Lean 4 theorem prover. Instead of directly translating traditional measure-theoretic definitions, we develop a geometric characterization based on probabilities over orthants and prove its equivalence to standard FSD conditions. This approach substantially reduces the formalization overhead while maintaining mathematical rigor, enabling tractable formal verification of complex economic decision problems.

The main contributions of this work are:
\begin{enumerate}[leftmargin=1.5em, itemsep=0.5em]
    \item \textbf{A Novel Geometric Framework for $N$-Dimensional FSD:} We develop and formalize in Lean 4 a characterization of FSD based on $N$-dimensional orthant indicator functions ($\indicatorURO$) and survival probabilities. This geometric approach provides an intuitive interpretation of FSD as the comparison of probabilities of exceeding arbitrary threshold vectors across all dimensions simultaneously.

    \item \textbf{Formally Verified Equivalence Proofs:} We provide formally verified proofs in Lean 4 for the equivalence between these geometric characterizations of FSD (i.e., higher survival probabilities) and the traditional expected utility characterizations. These proofs ensure that our geometric framework preserves the essential economic properties of stochastic dominance.

    \item \textbf{Demonstration of Enhanced Tractability in Formalization:} We illustrate how this geometric methodology simplifies the formalization process within Lean 4 compared to traditional analytical approaches, significantly reducing the technical overhead required for formal verification.

    \item \textbf{Elucidation of Direct Applicability to Economic Problems:} We explore the direct applicability of this formally verified geometric FSD framework to complex, multi-dimensional economic decision problems in portfolio selection, risk management, welfare analysis, and emerging areas such as data privacy and certified systems.
\end{enumerate}

By making the underlying mathematical structures more amenable to formal reasoning and verification, this geometric approach paves the way for increased rigor, reliability, and the development of certified decision-making tools across economic domains. It establishes a foundation for further development of verified libraries for economic analysis while simultaneously improving the accessibility of these formal methods to practitioners.

The rest of this paper is structured as follows: Section \ref{sec:lean_role} briefly discusses the role of the Lean 4 prover and the Mathlib library in this work. Section \ref{sec:fsd_1d} introduces the geometric formalization approach in the familiar one-dimensional setting. Section \ref{sec:geometric_framework_nd} extends this to the $N$-dimensional case, presenting our main geometric framework. Section \ref{sec:theorems_nd} states and proves key theorems about $N$-dimensional geometric FSD. Section \ref{sec:economic_applications} explores applications to economic problems, while Section \ref{sec:comparison_approaches} compares our approach with alternative formalization methods. Section \ref{sec:conclusion} concludes with a summary of our contributions, and Section \ref{sec:industry_impact} discusses broader industrial impacts.

\section{The Role of the Lean 4 Prover and Mathlib}\label{sec:lean_role}
The formalizations presented in this document were carried out using the Lean 4 interactive theorem prover \cite{Lean4Website}. Lean 4 is a functional programming language and a proof assistant based on dependent type theory, which provides a formal foundation for mathematical reasoning with computer-verified guarantees. Unlike traditional programming languages, Lean 4 enables users to state mathematical theorems and interactively develop formal proofs that are verified by the system's kernel. This ensures a level of mathematical rigor that surpasses what is typically achievable in conventional mathematical texts.

This work relies on Mathlib \cite{MathlibCommunity2020}, Lean's extensive, community-driven library of formalized mathematics. Mathlib provides foundational theories essential for this project, including:
\begin{itemize}[leftmargin=1.5em, itemsep=0.3em]
    \item Real number theory (\lean{$\mathbb{R}$}).
    \item Set theory, finite sets (\lean{Finset}), and interval notation (\lean{Set}, \lean{Icc}, \lean{Ioo}).
    \item Basic analysis concepts, although our geometric approach deliberately minimizes reliance on advanced integration theory.
    \item Foundations for probability and utility theory.
\end{itemize}
The process of formalization involves translating standard mathematical definitions and theorems into Lean's formal language and then interactively constructing proofs using Lean's tactic system. This requires precise specification of mathematical concepts, making implicit assumptions explicit, and building proofs in a step-by-step manner that the computer can verify. While this process is more demanding than traditional mathematical writing, it yields much stronger guarantees of correctness and can uncover subtle issues or implicit assumptions in established mathematical theories.

Our choice of Lean 4 over alternative proof assistants such as Coq \cite{CoqWebsite} or Isabelle/HOL \cite{IsabelleWebsite} was motivated by Lean's strong support for classical mathematics, its extensive mathematical library (Mathlib), and its growing adoption in formalizing economic theories. The geometric approach developed in this paper is particularly well-suited to Lean's capabilities, as it allows us to work with concrete, constructive definitions while leveraging classical reasoning where appropriate.

\section{Geometric Formalization of First-Order Stochastic Dominance in One Dimension}\label{sec:fsd_1d}
While our primary focus is on the $N$-dimensional case, we briefly outline the one-dimensional FSD formalization to introduce key concepts in a familiar setting. The geometric intuition, crucial for the multi-dimensional extension, is more readily understood here.

\subsection{Specialized Riemann-Stieltjes Integral for Indicator Functions}
The standard FSD equivalence theorem relates the condition $F(x) \leq G(x)$ for all $x$ (where $F, G$ are CDFs) to $E_F[u(X)] \geq E_G[u(X)]$ for all non-decreasing utility functions $u$. A common proof approach uses a specialized version of the Riemann-Stieltjes integral for indicator functions of the form $u(x) = \indicator{(x_0, \infty)}(x)$, where $\indicator{S}$ is the indicator function that equals 1 when the argument is in set $S$ and 0 otherwise. For such functions, the expected value can be directly computed as:
\[ E[u(X)] = \int_a^b \indicator{(x_0, \infty)}(x) \, d\text{Dist}(x) = \int_{x_0}^b 1 \, d\text{Dist}(x) = \text{Dist}(b) - \text{Dist}(x_0) = 1 - \text{Dist}(x_0), \]
assuming $x_0 \in (a,b)$.
Our formalization directly captures this specific calculation, bypassing the need for a general theory of Riemann-Stieltjes integration for this step.

\begin{definition}[Specialized Riemann-Stieltjes Integral for Indicator Functions]\label{def:Riemann-Stieltjes Integral_1D}
Let $u: \R \to \R$ be a utility function, $\text{Dist}: \R \to \R$ be a CDF, and $a, b \in \R$ with $a < b$. The specialized Riemann-Stieltjes integral of $u$ with respect to $\text{Dist}$ on $\Icc{a}{b}$ is defined as follows:

Let $P$ be the proposition that $u$ is an indicator function for some point $x_0$ in the open interval $(a,b)$:
\[ P \equiv \exists x_0 \in \Ioo{a}{b}, \forall x \in \Icc{a}{b}, u(x) = \begin{cases} 1 & \text{if } x > x_0 \\ 0 & \text{otherwise} \end{cases} \]
Using classical logic (specifically, \lean{Classical.propDecidable} to assert that $P$ is decidable, and \lean{Classical.choose} to extract the witness $x_0$ if $P$ holds)\footnote{Appendix \ref{classical reasoning} provides details on the use of classical reasoning in our formalization.}, we define:
\[ \lean{riemannStieltjesIntegral}(u, \text{Dist}, a, b) :=
   \begin{cases}
     1 - \text{Dist}(\text{Classical.choose}(P)) & \text{if } P \text{ holds} \\
     0             & \text{if } P \text{ does not hold (placeholder value)}
   \end{cases}
\]
\end{definition}

\begin{remark}[Formalization Note]
This definition is tailored for the FSD proof. It only yields the intended expected value when $u$ is precisely of the form $\indicator{(x_0, \infty)}(x)$ for some $x_0 \in \Ioo{a}{b}$. The use of \lean{Classical.propDecidable} and \lean{Classical.choose} makes this definition non-constructive, but sufficient for our theoretical equivalence proofs. While general Riemann-Stieltjes integration could be formalized, our specialized approach significantly reduces the formalization overhead while capturing the essential behavior needed for the FSD equivalence theorem.
\end{remark}

\begin{example}[Calculating the Specialized Integral]
Let $a=0, b=10$. Let $\text{Dist}(x) = x/10$ for $x \in [0,10]$ (a uniform CDF on $[0,10]$).
Consider the utility function $u(x) = \indicator{(3, \infty)}(x)$, i.e., $u(x)=1$ if $x>3$ and $0$ otherwise.
Here, the proposition $P$ is true, with $x_0 = 3 \in \Ioo{0}{10}$.
Then, $\lean{Classical.choose}(P)$ would yield $x_0=3$.
The integral is calculated as:
\[ \lean{riemannStieltjesIntegral}(u, \text{Dist}, 0, 10) = 1 - \text{Dist}(3) = 1 - (3/10) = 7/10. \]
This is $P(X > 3)$. If $u(x)$ was, for example, $u(x)=x^2$, then proposition $P$ would be false, and the integral definition would yield $0$.
\end{example}

To ensure this definition is well-behaved and that the $x_0$ chosen via \lean{Classical.choose} is unique (up to the behavior of $u$ on $\Icc{a}{b}$), we establish the following lemma.

\begin{lemma}[Uniqueness of Indicator Point (1D)]\label{lemma:uniqueness_1D}
Suppose $a < b$. If $u_1(x) = \indicator{(x_1, \infty)}(x)$ and $u_2(x) = \indicator{(x_2, \infty)}(x)$ agree for all $x \in \Icc{a}{b}$, where $x_1, x_2 \in \Ioo{a}{b}$, then $x_1 = x_2$.
\end{lemma}
\begin{proof}
Formal proof sketch provided in Appendix \ref{proof:lemma:uniqueness_1d}. The proof proceeds by contradiction, assuming $x_1 \neq x_2$ (e.g., $x_1 < x_2$) and evaluating the functions at a point between $x_1$ and $x_2$, such as their midpoint. This yields a contradiction since the indicator functions would produce different values at this point.
\end{proof}

\begin{remark}[Verification Perspective]
Lemma \ref{lemma:uniqueness_1D} is crucial for the logical consistency of Definition \ref{def:Riemann-Stieltjes Integral_1D}. It ensures that if a function $u$ matches the required indicator form $\indicator{(x_0, \infty)}$ for some $x_0 \in \Ioo{a}{b}$, then this $x_0$ is unique. Thus, regardless of which witness \lean{Classical.choose} selects, the result of our integral calculation will be correct. This obviates potential ambiguity that could arise from the non-constructive nature of the definition.
\end{remark}

\begin{lemma}[Integral Calculation for Indicator Functions (1D)]\label{lemma:integral_indicator_1D}
If $u(x) = \indicator{(x_0, \infty)}(x)$ for a specific $x_0 \in \Ioo{a}{b}$ (meaning $u(x)=1$ if $x>x_0$ and $0$ otherwise, for $x \in \Icc{a}{b}$), then $\lean{riemannStieltjesIntegral}(u, \text{Dist}, a, b) = 1 - \text{Dist}(x_0)$.
\end{lemma}
\begin{proof}
Formal proof sketch provided in Appendix \ref{proof:lemma:integral_indicator_1d}. The proof shows that the condition $P$ in Definition \ref{def:Riemann-Stieltjes Integral_1D} holds with $x_0$ as the explicit witness. Then, using Lemma \ref{lemma:uniqueness_1D}, we prove that \lean{Classical.choose(P)} must equal $x_0$, ensuring the integral evaluates to $1 - \text{Dist}(x_0)$.
\end{proof}

\begin{remark}[Purpose]
This lemma formally verifies that Definition \ref{def:Riemann-Stieltjes Integral_1D} correctly computes the expected value $1-\text{Dist}(x_0)$ when $u$ is indeed an indicator function of the form $\indicator{(x_0, \infty)}$. It provides the essential step for our subsequent FSD equivalence theorem, connecting the geometric properties of distributions ($1-\text{Dist}(x_0)$ is the probability of exceeding threshold $x_0$) with the expected utility framework.
\end{remark}

With these definitions and lemmas, we can state and prove the FSD equivalence theorem for this class of indicator functions.

\begin{theorem}[FSD Equivalence for Indicator Functions (1D)]\label{thm:fsd_iff_1D}
Given CDFs $F, G$ on $\Icc{a}{b}$ with $F(a)=G(a)=0$ and $F(b)=G(b)=1$.
The condition $F(x) \leq G(x)$ for all $x \in \Icc{a}{b}$ holds if and only if for all $x_0 \in \Ioo{a}{b}$,
\[ \lean{riemannStieltjesIntegral}(\indicator{(x_0, \infty)}, F, a, b) \geq \lean{riemannStieltjesIntegral}(\indicator{(x_0, \infty)}, G, a, b). \]
\end{theorem}
\begin{proof}
Formal proof sketch provided in Appendix \ref{proof:thm:fsd_iff_1d}.

($\Rightarrow$): Assume $F(x) \leq G(x)$ for all $x \in \Icc{a}{b}$. For any $x_0 \in \Ioo{a}{b}$, let $u_0(x) = \indicator{(x_0, \infty)}(x)$. By Lemma \ref{lemma:integral_indicator_1D}, the integral with respect to $F$ is $1-F(x_0)$, and the integral with respect to $G$ is $1-G(x_0)$. From our assumption, $F(x_0) \leq G(x_0)$, so $1-F(x_0) \geq 1-G(x_0)$, establishing the integral inequality.

($\Leftarrow$): Assume the integral inequality holds for all $x_0 \in \Ioo{a}{b}$. For any such $x_0$, applying Lemma \ref{lemma:integral_indicator_1D} to both sides gives $1 - F(x_0) \geq 1 - G(x_0)$, which simplifies to $F(x_0) \leq G(x_0)$. For the boundary cases $x_0=a$ and $x_0=b$, we use the given conditions $F(a)=G(a)=0$ and $F(b)=G(b)=1$, yielding $F(x) \leq G(x)$ for all $x \in \Icc{a}{b}$.
\end{proof}

\begin{remark}[Significance]
This theorem, formally verified in Lean, confirms the standard result that FSD (in terms of CDFs) is equivalent to higher expected utility for the specific class of "greater than $x_0$" indicator utility functions. This is a foundational step toward our $N$-dimensional geometric framework, establishing that comparing survival probabilities ($1-F(x_0)$) is equivalent to CDF comparisons in the one-dimensional case.
\end{remark}
\section{Geometric Framework for $N$-Dimensional Stochastic Dominance}\label{sec:geometric_framework_nd}
The true power and tractability of the geometric approach become particularly evident in $N$ dimensions. Here, traditional multi-dimensional integration and measure theory introduce significant complexity for formalization. Our geometric approach replaces this with a more direct characterization based on orthants and combinatorial principles.

\begin{figure}[H]
\centering
\begin{tikzpicture}[scale=0.85]
  \draw[->] (0,0) -- (5,0) node[right] {$X_1$};
  \draw[->] (0,0) -- (0,5) node[above] {$X_2$};

  \filldraw (2,3) circle (2pt) node[below right] {$x_0=(x_{0,1},x_{0,2})$};

  \fill[blue!10] (2,3) -- (5,3) -- (5,5) -- (2,5) -- cycle;
  \draw[dashed] (2,3) -- (2,5);
  \draw[dashed] (2,3) -- (5,3);

  \node at (3.5,4) {Upper-right orthant: $X_1>x_{0,1}, X_2>x_{0,2}$};
  \node at (3.5,1) {$P(X_1>x_{0,1}, X_2>x_{0,2}) = \survivalProbN(\text{Dist}, x_0, b)$};
  \node at (4,2) {$\indicatorURO(x_0,x) = 1$};
\end{tikzpicture}
\caption{Geometric interpretation of the upper-right orthant in 2D stochastic dominance. The survival probability measures the probability mass in the shaded region where all components exceed their threshold values.}
\label{fig:orthant_2d}
\end{figure}
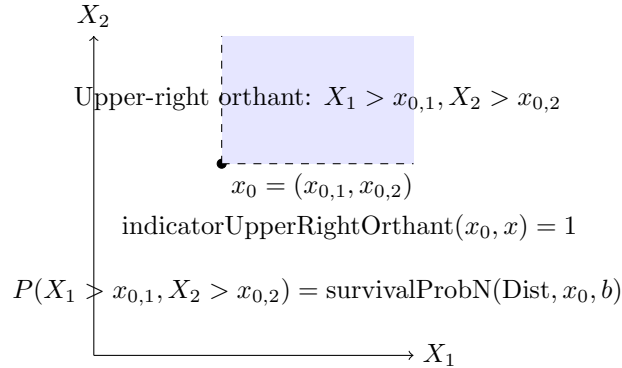

\begin{definition}[N-dimensional Vector of Reals]\label{def:Nd_Vector}
An $N$-dimensional vector of reals, denoted $\RVector(n)$ in our Lean formalization, is a function from $\Fin(n)$ to $\R$, where $\Fin(n)$ represents the finite set $\{0, 1, \dots, n-1\}$. This corresponds to the standard $\R^n$ but uses a function representation that facilitates formal reasoning in Lean.
\end{definition}
\begin{example}
If $n=3$, a vector $x \in \RVector(3)$ could be $x = (x_0, x_1, x_2)$, where $x(0)=x_0, x(1)=x_1, x(2)=x_2$. For instance, $v = (1.0, 2.5, -0.5)$ is in $\RVector(3)$.
\end{example}

\begin{definition}[Vector Relations]\label{def:vector relations}
Let $x, y \in \RVector(n)$.
\begin{itemize}
    \item \textbf{Componentwise less than ($<$)}: $x < y \iff \forall i \in \Fin(n), x(i) < y(i)$.
    \item \textbf{Componentwise less than or equal ($\leq$)}: $x \leq y \iff \forall i \in \Fin(n), x(i) \leq y(i)$.
    \item \textbf{Componentwise strictly greater than ($\allGt$)}: $\allGt(x, y) \iff \forall i \in \Fin(n), x(i) > y(i)$. (This is equivalent to $y < x$).
\end{itemize}
\end{definition}
\begin{example}
Let $x=(1,2), y=(3,4), z=(1,5) \in \RVector(2)$.
Then $x < y$ since $1<3$ and $2<4$.
However, $x \not< z$ because $x(0) \not< z(0)$ (i.e., $1 \not< 1$). But $x \leq z$ holds.
$\allGt(y,x)$ is true since $y(0)>x(0)$ and $y(1)>x(1)$.
\end{example}

\begin{definition}[N-dimensional Rectangles]\label{def:N-dimensional Rectangles}
Let $a, b \in \RVector(n)$.
\begin{itemize}
    \item \textbf{Closed N-dimensional Rectangle ($\Iccn$)}:
    $\Iccn(a,b) := \{x \in \RVector(n) \mid a \leq x \land x \leq b \}$.
    This means for all $i \in \Fin(n)$, $a(i) \leq x(i) \leq b(i)$. This represents the hyperrectangle $[a_0, b_0] \times \dots \times [a_{n-1}, b_{n-1}]$.
    \item \textbf{Open N-dimensional Rectangle ($\Ioon$)}:
    $\Ioon(a,b) := \{x \in \RVector(n) \mid a < x \land x < b \}$.
    This means for all $i \in \Fin(n)$, $a(i) < x(i) < b(i)$. This represents the hyperrectangle $(a_0, b_0) \times \dots \times (a_{n-1}, b_{n-1})$.
\end{itemize}
\end{definition}
\begin{example}
If $n=2$, $a=(0,0)$, $b=(1,2)$.
$\Iccn(a,b) = \{ (x_0, x_1) \mid 0 \leq x_0 \leq 1 \land 0 \leq x_1 \leq 2 \}$.
$\Ioon(a,b) = \{ (x_0, x_1) \mid 0 < x_0 < 1 \land 0 < x_1 < 2 \}$.
The point $(0.5, 1.5)$ is in $\Ioon(a,b)$ and $\Iccn(a,b)$. The point $(1,1)$ is in $\Iccn(a,b)$ but not $\Ioon(a,b)$.
\end{example}

\begin{definition}[Special Vector Constructions]\label{def:Special Vector Constructions}
Let $x_0, b, x \in \RVector(n)$, $s \subseteq \Fin(n)$ be a finite set of indices (formalized as \lean{Finset (Fin n)} in Lean), $j \in \Fin(n)$, and $\text{val} \in \R$.
\begin{itemize}
    \item \textbf{Mixed Vector ($\mixedVec$)}: $\mixedVec(x_0, b, s)$ is a vector where components at indices in $s$ are taken from $x_0$, and components at indices not in $s$ are taken from $b$.
    \[ (\mixedVec(x_0, b, s))(i) = \begin{cases} x_0(i) & \text{if } i \in s \\ b(i) & \text{if } i \notin s \end{cases} \]
    This construction is fundamental for defining the $N$-dimensional survival probability using the principle of inclusion-exclusion, as it allows us to specify points at the "corners" of orthants.
    \item \textbf{Replace Component ($\text{replace}$)}: $\text{replace}(x, j, \text{val})$ is a vector identical to $x$ except that its $j$-th component is replaced by $\text{val}$.
    \[ (\text{replace}(x, j, \text{val}))(i) = \begin{cases} \text{val} & \text{if } i = j \\ x(i) & \text{if } i \neq j \end{cases} \]
    This is a standard utility in vector manipulations.
    \item \textbf{Midpoint Vector ($\text{midpoint}$)}: $(\text{midpoint}(x,y))(i) = (x(i) + y(i))/2$ for all $i \in \Fin(n)$. This is a non-computable definition in Lean due to real division not being computable, but it serves as a useful theoretical construct for our proofs.
\end{itemize}
\end{definition}
\begin{example}[Mixed Vector]
Let $n=3$, $x_0 = (1,2,3)$, $b=(10,11,12)$. Let $s = \{0, 2\} \subseteq \Fin(3)$.
Then $\mixedVec(x_0, b, s)$ will take $x_0(0)$ and $x_0(2)$ and $b(1)$:
$(\mixedVec(x_0, b, s))(0) = x_0(0) = 1$
$(\mixedVec(x_0, b, s))(1) = b(1) = 11$
$(\mixedVec(x_0, b, s))(2) = x_0(2) = 3$
So, $\mixedVec(x_0, b, s) = (1, 11, 3)$.
\end{example}

\begin{definition}[Indicator Function for Upper-Right Orthant]\label{def:indicatorURO_ND}
The function $\indicatorURO : \RVector(n) \to \RVector(n) \to \R$ is defined as:
\[ \indicatorURO(x_0, x) := \begin{cases} 1 & \text{if } \allGt(x, x_0) \\ 0 & \text{otherwise} \end{cases} \]
This function is pivotal: it identifies whether a point $x$ strictly "dominates" a threshold point $x_0$ in all dimensions. The utility function $u(x) = \indicatorURO(x_0, x)$ represents a preference for outcomes that exceed threshold $x_0$ in every dimension, reflecting an economically meaningful preference structure.
\end{definition}
\begin{example}
Let $n=2$, $x_0 = (1,1)$.
If $x=(2,3)$, then $\allGt(x,x_0)$ is true (since $2>1$ and $3>1$), so $\indicatorURO(x_0, x) = 1$.
If $x=(0,4)$, then $\allGt(x,x_0)$ is false (since $0 \not> 1$), so $\indicatorURO(x_0, x) = 0$.
If $x=(2,1)$, then $\allGt(x,x_0)$ is false (since $1 \not> 1$), so $\indicatorURO(x_0, x) = 0$.
\end{example}

\begin{definition}[$N$-dimensional Survival Probability]\label{def:survivalProbN_ND}
The survival probability, denoted $\survivalProbN(\text{Dist}, x_0, b)$, for a joint distribution function $\text{Dist}: \RVector(n) \to \R$ (where $\text{Dist}(x) = P(X_0 \leq x_0, \dots, X_{n-1} \leq x_{n-1})$), threshold vector $x_0 \in \RVector(n)$, and upper bound vector $b \in \RVector(n)$, is defined as:

\[ \survivalProbN(\text{Dist}, x_0, b) := 1 - \sum_{s \in \mathcal{P}(\Fin(n)) \setminus \{\emptyset\}} (-1)^{|s|+1} \cdot \text{Dist}(\mixedVec(x_0, b, s)), \]
where $\mathcal{P}(\Fin(n))$ is the power set of indices $\Fin(n)$, and $\mixedVec(x_0, b, s)$ constructs a vector whose components at indices in $s$ are taken from $x_0$ and components at indices not in $s$ are taken from $b$.

This formula computes the probability that all components of a random vector $X$ exceed their respective thresholds in $x_0$:

\[ P(X > x_0) = P(X_0 > x_0(0), \dots, X_{n-1} > x_0(n-1)). \]

Using the principle of inclusion-exclusion, this can be equivalently expressed as:

\[ \survivalProbN(\text{Dist}, x_0, b) := \sum_{s \subseteq \Fin(n)} (-1)^{|s|} \cdot \text{Dist}(\mixedVec(x_0, b, s)). \]
\end{definition}
\begin{remark}[Formalization Note on $\survivalProbN$]
The definition of $\survivalProbN$ uses \lean{Finset.powerset} and \lean{Finset.sum} from Lean. The vector $b$ plays the role of the "upper anchor" for the probability calculation. For a standard CDF $F$ in two dimensions, with $b=(b_0,b_1)$ where $b_i$ represent upper bounds of the support:

Our $\mixedVec(x_0, b, s)$ constructs the arguments for $F$. For $s=\{0\}$, $\mixedVec(x_0,b,\{0\})=(x_{0,0}, b_1)$. For $s=\{1\}$, $\mixedVec(x_0,b,\{1\})=(b_0, x_{0,1})$. For $s=\{0,1\}$, $\mixedVec(x_0,b,\{0,1\})=(x_{0,0},x_{0,1})$.

The sum is:
$(-1)^{1+1} F(x_{0,0}, b_1) + (-1)^{1+1} F(b_0, x_{0,1}) + (-1)^{2+1} F(x_{0,0}, x_{0,1})$.
So $\survivalProbN = 1 - (F(x_{0,0},b_1) + F(b_0,x_{0,1}) - F(x_{0,0},x_{0,1}))$. This is indeed the standard formula for $P(X_0 > x_{0,0}, X_1 > x_{0,1})$.
This combinatorial definition is highly amenable to formal proof, as it relies on set theory and algebraic manipulation rather than analytic limits or measure theory for its basic properties.
\end{remark}

\begin{definition}[$N$-dimensional Riemann-Stieltjes Integral for Orthant Indicators]\label{def:riemannStieltjesND_ND}
The specialized Riemann-Stieltjes integral $\riemannStieltjesND(u, \text{Dist}, a, b, x_0, h_{x_0}, h_u)$ is defined for a function $u: \RVector(n) \to \R$, a distribution $\text{Dist}$, vectors $a, b, x_0 \in \RVector(n)$, a proof $h_{x_0}$ that $x_0 \in \Iccn(a,b)$, and a proof $h_u$ that $u(x) = \indicatorURO(x_0, x)$ for all $x \in \Ioon(a,b)$.
It evaluates to:
\[ \begin{cases} \survivalProbN(\text{Dist}, x_0, b) & \text{if } \exists x_0' \in \Iccn(a,b) \text{ s.t. } \forall x \in \Ioon(a,b), u(x) = \indicatorURO(x_0',x) \\ 0 & \text{otherwise} \end{cases} \]
The condition uses \Lean{Classical.propDecidable}. The definition in Lean directly uses $x_0$ from the arguments if the condition (checked via $h_u$ and $h_{x_0}$) holds.
\end{definition}

\begin{remark}[Verification Perspective]
Similar to the 1D case, this definition is tailored. The hypotheses $h_{x_0}$ and $h_u$ ensure that the function $u$ indeed corresponds to an orthant indicator defined by $x_0$ within the relevant domain. The definition takes advantage of the uniqueness of such $x_0$ (established in Lemma \ref{lem:uniqueness_ND}) to avoid unnecessary non-constructive operations when the threshold parameter $x_0$ is already known. This design enhances both the logical clarity and the tractability of subsequent proofs.
\end{remark}

\section{Key Theorems for $N$-Dimensional Geometric FSD}\label{sec:theorems_nd}
\subsection{Main Results}
The geometric framework allows for clear and formally verifiable statements about $N$-dimensional FSD, building upon the definitions above.

\begin{lemma}(Subset Relation between Rectangles)\label{lem:Ioo_subset_Icc_self_n}
For any $a, b \in \RVector(n)$, $\Ioon(a,b) \subseteq \Iccn(a,b)$.
\end{lemma}
\begin{proof}
Formal proof sketch provided in Appendix \ref{proof:lem:Ioo_subset_Icc_self_n}. This follows directly from $x < y \implies x \leq y$ applied componentwise.
\end{proof}
\begin{remark}[Formalization Note]
This is a foundational geometric property, easily proven in Lean by applying the definition of subset and vector relations. It's used frequently to ensure that points in an open rectangle are also contained in the corresponding closed rectangle, facilitating necessary type coercions in our proofs.
\end{remark}

\begin{lemma}[Uniqueness of Indicator Function Parameter ($N$D)]\label{lem:uniqueness_ND}
Let $a, b, x_1, x_2 \in \RVector(n)$ such that $\forall i, a(i) < b(i)$ (ensuring $\Ioon(a,b)$ is non-empty and $n$-dimensional).
Let $x_1 \in \Ioon(a,b)$ and $x_2 \in \Ioon(a,b)$. If for all $x \in \Ioon(a,b)$,
 $$ \indicatorURO(x_1, x) = \indicatorURO(x_2, x),$$
then $x_1 = x_2$.
\end{lemma}
\begin{proof}
Formal proof sketch provided in Appendix \ref{proof:lem:uniqueness_nd}. The proof is by contradiction. If $x_1 \neq x_2$, then they must differ in at least one component, say $x_1(j) \neq x_2(j)$. Assuming without loss of generality that $x_1(j) < x_2(j)$, we construct a point $z \in \Ioon(a,b)$ such that $x_1(j) < z(j) < x_2(j)$ and $z(i) > \max(x_1(i), x_2(i))$ for all $i \neq j$. This point $z$ satisfies $\allGt(z, x_1)$ but not $\allGt(z, x_2)$, so $\indicatorURO(x_1, z) = 1$ but $\indicatorURO(x_2, z) = 0$, contradicting the assumption that these functions are equal on $\Ioon(a,b)$.
\end{proof}
\begin{remark}[Verification Perspective]
This lemma is critical. It guarantees that if an orthant indicator utility function $u(x) = \indicatorURO(x_0,x)$ is specified over the domain $\Ioon(a,b)$, the parameter vector $x_0$ is unique. This justifies our use of classical reasoning in Definition \ref{def:riemannStieltjesND_ND}, ensuring that regardless of which $x_0'$ is chosen by \lean{Classical.choose}, the integral value $\survivalProbN(\text{Dist}, x_0', b)$ will be the same for any valid $x_0'$. In practice, our formalization avoids invoking \lean{Classical.choose} when $x_0$ is explicitly provided.
\end{remark}

\begin{lemma}[Integral for Indicator Function ($N$D)]\label{lem:integral_for_indicator_ND}
Let $a, b \in \RVector(n)$ with $\forall i, a(i) < b(i)$.
Let $u: \RVector(n) \to \R$, $\text{Dist}: \RVector(n) \to \R$, and $x_0 \in \RVector(n)$.
Assume $x_0 \in \Ioon(a,b)$ (hypothesis $h_{x_0{_mem}}$) and $\forall x \in \Iccn(a,b), u(x) = \indicatorURO(x_0, x)$ (hypothesis $h_{u{_def}}$).
Then, with $h_{x_0}' : x_0 \in \Iccn(a,b)$ (derived from $h_{x_0{_mem}}$ using Lemma \ref{lem:Ioo_subset_Icc_self_n}) and $h_u' : \forall x \in \Ioon(a,b), u(x) = \indicatorURO(x_0, x)$ (derived from $h_{u{_def}}$),
$$\riemannStieltjesND(u, \text{Dist}, a, b, x_0, h_{x_0}', h_u') = \survivalProbN(\text{Dist}, x_0, b).$$
\end{lemma}
\begin{proof}
Formal proof sketch provided in Appendix \ref{proof:lem:integral_for_indicator_nd}. The proof directly applies Definition \ref{def:riemannStieltjesND_ND}. The hypotheses $h_{x_0{_mem}}$ and $h_{u{_def}}$ establish that $u$ is indeed an orthant indicator function defined by $x_0$. The definition of $\riemannStieltjesND$ then evaluates to $\survivalProbN(\text{Dist}, x_0, b)$ as required.
\end{proof}
\begin{remark}[Purpose]
This lemma is the cornerstone connecting the specialized integral definition to the geometric/combinatorial survival probability. It formally verifies that, for the specific class of orthant indicator utility functions, our specialized Riemann-Stieltjes integral correctly computes the survival probability. This establishes a bridge between the expected utility framework and the geometric interpretation of FSD through survival probabilities.
\end{remark}

\begin{theorem}[FSD Equivalence for Indicator Functions ($N$D)]\label{thm:fsd_equivalence_ND}
Let $F, G: \RVector(n) \to \R$ be joint distribution functions. Let $a,b \in \RVector(n)$ with $\forall i, a(i) < b(i)$.
Then, the condition
$$(\forall x_0 \in \Ioon(a,b), \survivalProbN(F, x_0, b) \geq \survivalProbN(G, x_0, b))$$
is equivalent to
$$(\forall (x_0 \in \Ioon(a,b)) (h_{x_0{_mem}} : x_0 \in \Ioon(a,b)),$$
$$  \riemannStieltjesND(\indicatorURO(x_0), F, a, b, x_0, h_{x_0}', h_u')$$
$$  \geq \riemannStieltjesND(\indicatorURO(x_0), G, a, b, x_0, h_{x_0}', h_u'),$$
where $h_{x_0}'$ is $x_0 \in \Iccn(a,b)$  and $h_u'$ is $\forall x \in \Ioon(a,b), u(x) = \indicatorURO(x_0, x)$ indicate the necessary proof arguments $h_{x_0}'$ and $h_u'$ as in Lemma \ref{lem:integral_for_indicator_ND}.
\end{theorem}
\begin{proof}
Formal proof sketch provided in Appendix \ref{proof:thm:fsd_equivalence_nd}.
The proof follows directly from Lemma \ref{lem:integral_for_indicator_ND}.
($\Rightarrow$): Assume $\forall x_0 \in \Ioon(a,b), \survivalProbN(F, x_0, b) \geq \survivalProbN(G, x_0, b)$. For any $x_0 \in \Ioon(a,b)$, let $u_{x_0}(x) = \indicatorURO(x_0, x)$. By Lemma \ref{lem:integral_for_indicator_ND}, the specialized Riemann-Stieltjes integrals with respect to $F$ and $G$ are precisely $\survivalProbN(F, x_0, b)$ and $\survivalProbN(G, x_0, b)$. The inequality between the survival probabilities directly translates to the inequality between the integrals.

($\Leftarrow$): Assume the integral inequality for all $x_0 \in \Ioon(a,b)$ and their corresponding indicator functions. By Lemma \ref{lem:integral_for_indicator_ND}, this translates directly to $\survivalProbN(F, x_0, b) \geq \survivalProbN(G, x_0, b)$ for all $x_0 \in \Ioon(a,b)$, completing the proof.
\end{proof}

\begin{remark}[Central Result for Geometric FSD]
This theorem is the central result for our geometric formalization of $N$-dimensional FSD for orthant indicator functions. It formally establishes that dominance in terms of higher expected utility for orthant indicator utility functions is equivalent to higher survival probabilities in all upper-right orthants. This provides a clear geometric interpretation of FSD in $N$ dimensions: one distribution dominates another if and only if it has a higher probability of exceeding any given threshold vector in all dimensions simultaneously.
\end{remark}

\subsection{Extension to General Non-Decreasing Utility Functions}\label{subsec:extension_general}
The formal verification presented thus far specifically addresses orthant indicator functions of the form $\indicatorURO(x_0, x)$. However, in Sections 6 and 9, we make broader claims that the results extend to "all investors with non-decreasing utility functions" and "all policy evaluators with non-decreasing utility functions." This extension, while not formally verified in our Lean implementation, rests on well-established mathematical principles in stochastic dominance theory \cite{Levy2015}\cite{ShakedShanthikumar2007}.

\begin{proposition}[Extension to Non-Decreasing Utility Functions]\label{prop:extension}
If $F$ and $G$ are two joint distribution functions such that $\survivalProbN(F, x_0, b) \geq \survivalProbN(G, x_0, b)$ for all $x_0 \in \Ioon(a,b)$, then $E_F[u(X)] \geq E_G[u(X)]$ for all non-decreasing utility functions $u: \RVector(n) \to \R$.
\end{proposition}

The mathematical foundations for this extension include:

\begin{enumerate}[leftmargin=1.5em, itemsep=0.5em]
    \item \textbf{Approximation Theory:} Any non-decreasing utility function can be approximated arbitrarily closely by a positive linear combination of upper-right orthant indicator functions. This is analogous to how step functions can approximate continuous functions in one dimension.

    \item \textbf{Linearity of Expectation:} If $$E_F[\indicatorURO(x_0,X)] \geq E_G[\indicatorURO(x_0,X)]$$ for all $x_0$, then by linearity of the expectation operator, this inequality also holds for any positive linear combination of these indicator functions.

    \item \textbf{Limit Arguments:} Taking limits (under appropriate regularity conditions), the dominance relation extends to all non-decreasing utility functions that can be approximated by such combinations.
\end{enumerate}

In formal mathematics terms, the orthant indicator functions form a "generating class" for the space of non-decreasing utility functions, in the sense that their linear spans and limits can represent any such function \cite{RussellSeo1989}.

\begin{remark}[Formalization Scope]
While formalizing the full extension to all non-decreasing utility functions would be a valuable contribution, it would require significant additional machinery in Lean, including:
\begin{itemize}
    \item Development of measure-theoretic function approximation theorems
    \item Formalization of convergence theorems for expectations
    \item Construction of limit operations for utility function sequences
\end{itemize}
These extensions remain as promising directions for future formal verification work, building upon the geometric foundation established in this paper.
\end{remark}

For the economic applications discussed in subsequent sections, the indicator function characterization provides a clear, testable condition through survival probability comparisons. The extension to all non-decreasing utility functions then follows as a "free theorem" from standard stochastic dominance theory, making the geometric approach both theoretically complete and practically applicable.

\section{Economic Applications: A Geometric Perspective}\label{sec:economic_applications}
The geometric framework for $N$-dimensional FSD, formalized and verified in Lean 4, offers a robust and notably tractable foundation for analyzing a range of multi-dimensional economic decision problems. In this section, we explore concrete applications across several domains and highlight how the formalization enhances these analyses.

\subsection{Portfolio Selection}
In contemporary finance, portfolio selection extends beyond the classic mean-variance analysis of Markowitz \cite{Markowitz1952}. Investors often evaluate portfolios based on multiple criteria, such as expected returns across different economic scenarios, risk measures across various time horizons, or environmental, social, and governance (ESG) metrics alongside financial performance.
Let a portfolio $P$'s outcome be represented by an $N$-dimensional random vector $X_P = (X_1, X_2, \dots, X_N)$, where each $X_i$ is a desirable attribute.

A portfolio $A$ exhibits first-order stochastic dominance over portfolio $B$ under our geometric framework if, for any target outcome vector $x_0 = (x_{0,1}, \dots, x_{0,N})$, the probability that portfolio $A$ exceeds all targets simultaneously is at least as high as the probability for portfolio $B$:
$P(X_A > x_0) \geq P(X_B > x_0)$.

\begin{example}[Multi-Attribute Portfolio Choice]
Consider two portfolios with bivariate distributions of (return, downside protection). Portfolio $A$ has a higher probability than portfolio $B$ of simultaneously achieving any given threshold combination of return and downside protection. By our geometric FSD framework, portfolio $A$ dominates portfolio $B$, making it preferable for all investors with non-decreasing utility functions over both attributes.
\end{example}

\textbf{Benefits and Enhanced Applications from Geometric Formalization:}
\begin{itemize}
    \item \textbf{Rigor for Complex Criteria:} The geometric framework provides a formal, verifiable basis for portfolio selection involving multiple, possibly correlated criteria without requiring explicit utility functions. This is particularly valuable for institutional investors who must justify decisions across multiple objectives.
    \item \textbf{Verifiable Robo-Advising:} Automated investment advisors can implement FSD checks on portfolios with formal guarantees, enhancing trust in recommendation algorithms.
    \item \textbf{Executable FSD Checkers for Empirical Distributions:} The constructive nature of the geometric proofs, especially for distributions represented as step-functions (empirical CDFs from historical data), enables the development of verified algorithms that can directly check for FSD between portfolios.
        \begin{itemize}
            \item \textbf{Application Example (Portfolio Screening):} A financial institution can use Lean to develop a verified program that decides whether an empirical CDF of returns from a candidate portfolio is dominated by an existing portfolio. The formal guarantees ensure that the screening process correctly implements FSD checks.
        \end{itemize}
    \item \textbf{Automated Verification for Parametric Models with Satisfiability Modulo Theories (SMT) Solvers:} If portfolio returns are modeled by parametric distributions, FSD conditions might translate into algebraic inequalities that can be verified through SMT solvers integrated with Lean.
        \begin{itemize}
            \item \textbf{Application Example:} A firm designs a structured product whose multi-attribute payoff depends on parameters $\alpha, \beta$. They can formally verify in Lean (potentially with SMT solver support) that for certain parameter ranges, the product's payoff distribution dominates a benchmark, providing formal guarantees to clients.
        \end{itemize}
    \item \textbf{Scalability for High-Dimensional Attributes:} The geometric approach, relying on combinatorics (inclusion-exclusion) rather than complex multi-dimensional integration theory (Fubini/Tonelli theorems), scales more effectively to higher dimensions in formal verification contexts.
\end{itemize}
This approach complements works like \cite{DentchevaRuszczynski2006} on stochastic dominance constraints in portfolio optimization, by providing a framework for formally verifying the underlying dominance relations.

\subsection{Risk Management}
Financial institutions manage multifaceted risks, represented as vectors $X = (X_1, \dots, X_N)$ of desirable outcomes (e.g., $X_1$=capital adequacy, $X_2$=liquidity ratio). Strategy $S_A$ FSDs $S_B$ if for all threshold vectors $x_0$, $P(X_{S_A} > x_0) \geq P(X_{S_B} > x_0)$.

\begin{example}[Comparing Hedging Strategies]
A bank compares two hedging strategies for managing interest rate and credit risks. Strategy $A$ stochastically dominates strategy $B$ in our geometric framework, meaning that for any target threshold combinations of interest rate protection and credit risk mitigation, strategy $A$ has a higher probability of simultaneously exceeding both thresholds. This provides a clear justification for preferring strategy $A$ without requiring specific utility functions over these risk dimensions.
\end{example}

\textbf{Benefits and Enhanced Applications from Geometric Formalization:}
\begin{itemize}
    \item \textbf{Unambiguous Model Comparison:} The geometric FSD condition provides clear, interpretable criteria for comparing risk management models across multiple risk dimensions.
    \item \textbf{Regulatory Confidence and Certified Compliance:} Regulators might demand formal proof that a new financial product or risk management strategy FSDs a baseline scenario across multiple risk dimensions.
        \begin{itemize}
            \item \textbf{Application Example:} A bank develops a new internal model for assessing operational risk across $N$ categories. To gain regulatory approval, they provide a Lean-verified proof that their model's risk distribution FSDs the standardized approach across all relevant risk dimensions, demonstrating superior risk management with mathematical certainty.
        \end{itemize}
    \item \textbf{FSD Constraints in Optimization Solvers:} Risk management often involves optimizing resource allocation (e.g., capital, hedging instruments) subject to risk constraints. FSD conditions on multivariate outcomes can be incorporated as verified constraints in optimization problems.
        \begin{itemize}
            \item \textbf{Application Example:} An insurer wants to design a reinsurance program by choosing among various contracts. The goal is to minimize reinsurance cost while ensuring that the post-reinsurance risk profile FSDs a regulatory benchmark across multiple risk categories. The geometric formulation enables more direct incorporation of these constraints in the optimization model.
        \end{itemize}
    \item \textbf{Certified Monte Carlo for Stress Testing:} When evaluating risk under stress scenarios using Monte Carlo simulations, the geometric FSD framework aligns directly with counting simulated outcomes in specific orthants, making verification more tractable.
        \begin{itemize}
            \item \textbf{Application Example:} A bank performs $M$ Monte Carlo simulations for two different investment portfolios under a stress scenario, yielding $M$ vectors of $N$ P\&L figures for each portfolio. A Lean-verified checker can confirm whether one portfolio stochastically dominates the other by directly implementing the geometric survival probability comparisons.
        \end{itemize}
\end{itemize}

\subsection{Welfare Analysis and Policy Evaluation}
Societal well-being or policy impact is often multi-dimensional ($W = (W_1, \dots, W_N)$). Policy $A$ FSDs policy $B$ if $\survivalProbN(F_A, w_0, b) \geq \survivalProbN(F_B, w_0, b)$ for any target welfare vector $w_0$, meaning $A$ has a higher probability than $B$ of simultaneously achieving all welfare thresholds.

\begin{example}[Evaluating Social Programs]
Two healthcare reforms are being compared based on their impacts on three dimensions: access to care, quality of care, and cost reduction. Reform $A$ FSDs reform $B$ in our geometric framework if the probability of simultaneously achieving any given threshold combination across all three dimensions is higher under reform $A$ than reform $B$. This provides a robust basis for policy selection without requiring explicit trade-offs between these objectives.
\end{example}

\textbf{Benefits and Enhanced Applications from Geometric Formalization:}
\begin{itemize}
    \item \textbf{Transparent Policy Choice:} The geometric condition $P(W > w_0)$ is visual and directly interpretable to policymakers as the probability of exceeding all welfare targets simultaneously.
    \item \textbf{Robustness to Utility Specification:} The FSD criterion is valid for all policy evaluators with non-decreasing utility functions over the welfare dimensions, avoiding contentious assumptions about specific social welfare functions.
    \item \textbf{Rigorous Impact Assessment for Interacting Policies:} Our framework accommodates joint distributions where welfare dimensions may be correlated, capturing interaction effects between policy components that might be missed in dimension-by-dimension analyses.
    \item \textbf{Pedagogical Clarity:} The geometric approach, linking FSD to probabilities of exceeding targets in hyperrectangles, is more accessible to students and policymakers who may know calculus but not measure theory, broadening the potential audience for formal methods in policy analysis.
    \item \textbf{Verified Simulation of Policy Impacts:} Similar to stress testing, if the impact of social policies is estimated via agent-based models or microsimulations yielding multi-dimensional outcome distributions, FSD can be assessed directly from these simulated datasets with formal guarantees.
\end{itemize}

\subsection{New Application Area: Formal Guarantees in Data Privacy}
The concept of "dominance" can be extended to information leakage in data privacy mechanisms, such as those aiming for differential privacy \cite{Dwork2006}. Consider $N$ different types of sensitive information. Let $X_i$ represent the information leakage (a value to be minimized) for type $i$ under privacy mechanism $M_A$, and $Y_i$ the leakage under mechanism $M_B$. If for all threshold vectors $x_0$, $P(X < x_0) \ge P(Y < x_0)$, then mechanism $A$ provides superior privacy protection across all dimensions.
Alternatively, if $X_i$ represents the utility preserved for attribute $i$ under privacy mechanism $M_A$, and $Y_i$ under $M_B$, then $M_A$ FSDs $M_B$ if $P(X > x_0) \ge P(Y > x_0)$ for all utility thresholds $x_0$, indicating better utility preservation.

\textbf{Benefits of Geometric Formalization:}
\begin{itemize}
    \item \textbf{Precise Multi-Attribute Privacy Guarantees:} The $N$-dimensional geometric framework can precisely define and verify privacy guarantees across multiple types of data or queries simultaneously, advancing beyond single-metric privacy analyses.
    \item \textbf{Reduced Technical Overhead for Verification:} Proving such multi-dimensional privacy properties without heavy measure-theoretic machinery makes formal verification more accessible to privacy researchers, potentially accelerating the development and verification of privacy-preserving algorithms.
\end{itemize}

\subsection{New Application Area: Certified Libraries and Embedded Systems}
The reduced dependency on heavy measure theory libraries makes the geometric FSD formalization suitable for inclusion in certified numerical libraries or systems where code size and auditability are crucial.

\textbf{Benefits of Geometric Formalization:}
\begin{itemize}
        \item \textbf{Lightweight Certified Components:} A verified FSD checker based on geometric principles can be a small, self-contained module. This is advantageous for embedded systems in finance (e.g., trading devices) and safety-critical applications where code verification is essential.
    \item \textbf{Foundation for Verified Complex Systems:} Verified FSD can be a building block in larger verified systems. For example, a formally verified dynamic programming solver could use this framework to ensure that decisions satisfy stochastic dominance requirements without requiring the full machinery of measure theory.
\end{itemize}
\section{Comparison with Analytical and Other Formal Approaches}\label{sec:comparison_approaches}
The traditional mathematical treatment of stochastic dominance, particularly in multiple dimensions, is deeply rooted in measure theory and the calculus of multi-dimensional integration \cite{Levy2015}. While these analytical approaches provide theoretical rigor, they introduce significant formalization challenges in interactive theorem proving environments.

Our geometric approach, as presented in this paper, offers a distinct and, for the specific goal of verifying FSD and enabling its direct applications, a more tractable alternative for formalization within Lean 4 and similar systems. The key differences and advantages include:

\begin{itemize}[leftmargin=1.5em, itemsep=0.5em]
    \item \textbf{Simplified Mathematical Primitives:} Instead of confronting general integrals and arbitrary measurable sets from the outset, our framework focuses on geometrically intuitive objects: upper-right orthants, hyperrectangles, and inclusion-exclusion calculations on discrete sets. This significantly reduces the theoretical prerequisites for formalization while maintaining the essential mathematical properties needed for economic applications.

    \item \textbf{Lower Technical Overhead in Formalization:} A direct consequence is a significant reduction in the formalization overhead.
        \begin{itemize}
            \item There is no direct need to define or manipulate $\sigma$-algebras, prove measurability for numerous functions and sets, or reason about almost-everywhere equalities for the core FSD theorems.
            \item Such self-contained formalizations tend to compile faster and place less strain on Lean's kernel. They are also often easier to maintain and adapt across future versions of Mathlib's evolving integration and measure theory libraries.
        \end{itemize}

    \item \textbf{Enhanced Automation and SMT-Friendliness:} Proofs within the geometric framework frequently reduce to verifying systems of linear (or sometimes bilinear) real arithmetic inequalities.
        \begin{itemize}
            \item Lean's built-in tactics like `linarith`, `polyrith`, `ring`, and `positivity` are highly effective for such goals.
            \item Furthermore, these types of arithmetic problems are often well-suited for external SMT solvers \cite{deMoura2008}, which can be integrated with Lean through frameworks like `smt\_tactic`.
            \item This contrasts sharply with measure-theoretic proofs, where automation often struggles with goals like "show this set is measurable" or "show this function is integrable," which require extensive manual guidance and domain-specific expertise.
        \end{itemize}

    \item \textbf{More Direct Path to $N$-Dimensional Scalability:} The traditional route to multivariate FSD requires formalizing product $\sigma$-algebras and Fubini's or Tonelli's theorems to handle multiple integrals. Our combinatorial approach using inclusion-exclusion principles scales to $N$ dimensions with comparatively less formalization overhead, making it practical to handle higher-dimensional problems.

    \item \textbf{Potential for Constructive Proofs and Code Extraction:} Because geometric proofs often rely on explicit inequalities and constructions on rectangles or intervals, the conclusions can frequently be stated in more computationally tractable ways.
        \begin{itemize}
            \item This opens the possibility of extracting executable code directly from the verified theorems (e.g., a program that decides FSD for distributions represented by step-functions or empirical CDFs).
            \item In contrast, some measure-theoretic results might be non-constructive or rely on classical axioms in ways that make direct code extraction more challenging or less meaningful.
        \end{itemize}
        
     \item \textbf{Lean-Specific Implementation Insights:} Working in Lean 4 specifically, We've found that:
        \begin{itemize}
            \item The \lean{Finset} operations used for survival probability calculations align well with Lean's computation capabilities and tactic-based automation
            \item The \lean{Classical.propDecidable} and \lean{Classical.choose} constructs provide a pragmatic way to handle the existence requirements while maintaining logical consistency
            \item Lean's dependent type system allows us to express the complex geometric conditions with precise types that carry proof information (like our \lean{riemannStieltjesND} function that takes proof arguments)
        \end{itemize}  
    
\end{itemize}

Compared to other formal verification efforts in economics or finance that might aim to formalize stochastic dominance by directly translating the standard analytical definitions (relying heavily on Mathlib's integration theory), our geometric approach offers a more modular and focused pathway. It establishes the essential properties of FSD needed for economic applications without requiring the formalization of the full measure-theoretic foundations. This strategic choice enhances tractability while maintaining mathematical rigor for the specific goal of enabling formal analysis of multi-dimensional economic decision problems.

\section{Conclusion}\label{sec:conclusion}
This paper has introduced and demonstrated a novel geometric pathway to the formalization and verification of $N$-dimensional first-order stochastic dominance using the Lean 4 theorem prover. By strategically reformulating the traditional analytical characterizations of FSD into a geometric framework based on survival probabilities in upper-right orthants, we have achieved a significantly more tractable formalization while maintaining the essential mathematical properties needed for economic analysis.

The central message of this work, substantiated by our formal developments in Lean 4, is that \textbf{geometric methods can provide a significantly more tractable path to formalizing complex probabilistic concepts in economics}, particularly those involving multi-dimensional risk and welfare comparisons. This approach reduces the formalization overhead without compromising mathematical rigor, bridging the gap between theoretical economic concepts and formal verification technologies.

The benefits of this geometric approach, particularly from the perspective of formal verification and practical application, are manifold:
\begin{itemize}[leftmargin=1.5em, itemsep=0.5em]
    \item \textbf{Enhanced Formal Tractability and Maintainability:} The geometric definitions simplify the translation of mathematical concepts into Lean 4 and streamline the proof development process. They also reduce dependencies on evolving libraries for advanced measure theory, enhancing long-term maintainability.
     \item \textbf{Clarity, Intuition, and Pedagogical Transparency:} Geometric conditions, such as comparing probabilities of exceeding targets in specific orthants ($P(X > x_0)$), often provide a more intuitive and accessible framework for practitioners than traditional analytical formulations involving multiple integrals or complex measure-theoretic constructs.
    \item \textbf{Foundation for Verifiable Economic Models and Certified Tools:} This work provides formally verified building blocks for constructing more complex economic models and decision-support tools with mathematical guarantees. The reduced formalization overhead makes it practical to incorporate these verified components into larger systems.
    \item \textbf{Improved Automation in Proofs:} The reduction of many proof goals to real arithmetic makes them amenable to powerful automated tactics in Lean (`linarith`, `polyrith`, `ring`, `positivity`) and integration with external SMT solvers, increasing the efficiency of formal verification efforts.
\end{itemize}

Future research can extend this geometric framework in several promising directions. This includes formalizing higher-order stochastic dominance using related geometric ideas (e.g., based on integrals of survival functions), developing verified algorithms for checking stochastic dominance between empirical distributions, and extending the approach to more specialized forms of stochastic dominance relevant to specific economic applications. The combinatorial structure of our definitions also suggests potential connections to algorithmic game theory and computational economics that warrant further exploration.

By demonstrating the feasibility and advantages of a geometric approach to a cornerstone concept in decision theory, this work aims to encourage further exploration of how strategic mathematical reformulations can enhance the tractability of formal verification in economics. The resulting formally verified theorems provide a foundation for more reliable and transparent economic analysis in high-stakes domains where rigorous guarantees are increasingly essential.

\section{Broader Industrial Impact and Applications}\label{sec:industry_impact}

While the formal verification of $N$-dimensional stochastic dominance provides significant theoretical contributions to mathematical economics, its applications extend into various industries where multi-dimensional decision-making under uncertainty is critical. This section explores the potential impact of our work beyond academic settings.

\subsection{Impact on Traditional Industries}\label{subsec:trad_industries}

The geometric framework for $N$-dimensional FSD formalized in this paper offers practical value to several industries:

\begin{itemize}[leftmargin=1.5em, itemsep=0.5em]
    \item \textbf{Financial Services:} Certified multi-dimensional portfolio comparisons provide mathematical guarantees for investment decisions. Asset management firms can implement verified algorithms to identify stochastically dominant strategies across multiple risk-return metrics, enhancing client trust through rigorously verified selection methodologies.

    \item \textbf{FinTech and Algorithmic Trading:} Our framework enables implementation of verified comparison operators in trading systems. Robo-advisor platforms can build trustworthy multi-criteria recommendation engines with formal guarantees that their suggestions are optimal for all clients with non-decreasing utility functions over relevant attributes.

    \item \textbf{Insurance Industry:} Property \& casualty insurers can leverage this formalization for modeling multi-dimensional catastrophe risks. Reinsurance companies can formally verify complex treaty structures using verified stochastic dominance checkers, ensuring optimal risk transfer across multiple peril categories simultaneously.

    \item \textbf{Regulatory Technology:} Compliance solution providers can build verified tools for regulatory reporting with FSD-based certifications. Model validation teams can adopt formal methods to verify that internal risk models dominate standard regulatory approaches, potentially justifying reduced capital requirements.

    \item \textbf{Manufacturing and Supply Chain:} Multi-parameter production processes under uncertainty can be optimized with verified guarantees. Supply chain risk management systems can implement certified resilience metrics by analyzing stochastic dominance across multiple disruption scenarios and resource constraints.

    \item \textbf{Healthcare Economics:} Multi-dimensional risk-benefit analysis of treatments can be formalized, enabling more reliable medical decision support systems. Resource allocation algorithms for healthcare systems can incorporate verified stochastic dominance checks for comparing intervention strategies across multiple health outcomes and cost dimensions.
\end{itemize}

The key value proposition across these sectors is transitioning from empirical or approximation-based approaches to decision systems with formal guarantees—particularly critical in high-stakes domains where errors can have significant financial, regulatory, or human consequences.

\subsection{Transformative Potential for AI Systems}\label{subsec:ai_impact}

The formalization of N-dimensional stochastic dominance has particularly promising applications in artificial intelligence:

\begin{itemize}[leftmargin=1.5em, itemsep=0.5em]
    \item \textbf{Formally Verified Decision-Making:} AI systems can incorporate verified preference ordering in multi-attribute decision frameworks. This provides mathematical guarantees that algorithmic decisions respect stochastic dominance principles, even under uncertainty. The orthant indicator representation aligns naturally with AI systems that compute probabilities over regions of feature space.

    \item \textbf{Multi-objective Reinforcement Learning (RL):} Our framework enables certified algorithms for comparing multi-dimensional reward distributions, formalizing correctness of Pareto front approximations in reinforcement learning. As a Lean prover, I've found that the geometric approach to FSD yields particularly clean specifications for verification of multi-objective RL algorithms.

    \item \textbf{Verified Fairness and Robustness:} Formal verification of fairness properties across multiple stochastic attributes becomes tractable. AI systems can maintain provable fairness guarantees across distributions of outcomes affecting different demographic groups, with formal verification of these properties in Lean.

    \item \textbf{Enhanced Uncertainty Quantification:} The geometric approach enables verified propagation of multi-dimensional uncertainties through AI pipelines. This includes certified comparison of output distributions from different model architectures, providing formal guarantees about their relative performance characteristics.

    \item \textbf{Foundation Models and Reasoning:} Next-generation AI can incorporate verified reasoning about probabilistic outcomes and multi-criteria preferences. This provides building blocks for symbolic AI systems that make provably correct inferences about stochastic dominance relations in complex domains.

    \item \textbf{AI Risk Assessment Frameworks:} Multi-dimensional risk analysis with formal guarantees becomes implementable, allowing verified risk aggregation across multiple AI failure modes. This contributes to the development of provably safer AI systems by formalizing comparative risk analysis across multiple safety dimensions.r.
\end{itemize}

The geometric formalization approach is particularly well-suited for AI applications due to its reduced dependency on heavy measure theory, making it more amenable to lightweight implementation in constrained environments like edge AI systems or verified runtime monitors. The Lean 4 implementation offers a path to extracting verified code that can be integrated into AI decision pipelines.

The formal verification of N-dimensional stochastic dominance using the geometric approach represents a significant step toward more reliable decision-making systems across multiple industries, with particular promise for enhancing the trustworthiness of AI systems that make high-stakes decisions under uncertainty.

\newpage
\appendix
\section{Appendix: Lean 4 Implementations and Proof Sketches}\label{sec:appendix}

Throughout this appendix, Lean 4 code snippets are illustrative of the formal definitions and theorems discussed in the main text. Proof sketches aim to convey the logical structure of the formal proofs while abstracting away some implementation details.

\subsection{Lean 4 Code Snippets}
The following snippets represent key definitions and theorem statements as formalized in Lean 4, utilizing the Mathlib library.

\textbf{Lean 4 Definition \ref{def:Riemann-Stieltjes Integral_1D} (Specialized Riemann-Stieltjes Integral for Indicators (1D)):}
\begin{lstlisting}[caption={Specialized Riemann-Stieltjes Integral for Indicators (1D)}, label={lst:rs_integral_1d}]
noncomputable def riemannStieltjesIntegral (u : \R \to \R) (Dist : \R \to \R) (a b : \R) : \R :=
  let P : Prop := \exists x\0 \in Ioo a b, \forall x \in Icc a b, u x = if x > x\0 then 1 else 0
  haveI : Decidable P := Classical.propDecidable P -- Asserts decidability using classical logic
  if hP : P then -- hP is a proof that P holds
    let x\0_witness := Classical.choose hP -- Extracts the witness x\0 from the proof hP
    1 - Dist x\0_witness
  else
    0 -- Placeholder if u is not of the specified indicator form
\end{lstlisting}

\textbf{Lean 4 Lemma \ref{lemma:uniqueness_1D} (Uniqueness of Indicator Point (1D)):}
\begin{lstlisting}[caption={Uniqueness of Indicator Point (1D)}, label={lst:1d_uniqueness}]
lemma uniqueness_of_indicator_x0_1D {a b x1 x2 : \R} (hab : a < b)
    (hx1_mem : x1 \in Ioo a b) (hx2_mem : x2 \in Ioo a b)
    (h_eq_fn : \forall x \in Icc a b, (if x > x1 then (1 : \R) else (0 : \R)) =
                                     (if x > x2 then (1 : \R) else (0 : \R))) :
    x1 = x2 :=
  sorry -- Formal proof omitted in snippet, sketch in Appendix \ref{proof:lemma:uniqueness_1d}
\end{lstlisting}

\textbf{Lean 4 Lemma \ref{lemma:integral_indicator_1D} (Integral Calculation for Indicator Functions (1D)):}
\begin{lstlisting}[caption={Integral Calculation for Indicator Functions (1D)}, label={lst:1d_Integral_Calculation_for_Indicator_Functions}]
lemma integral_for_indicator_1D {a b : \R} (hab : a < b) {u : \R \to \R} {Dist : \R \to \R}
    {x0 : \R} (hx0_mem : x0 \in Ioo a b)
    (h_u_def : \forall x \in Icc a b, u x = if x > x0 then 1 else 0) :
    riemannStieltjesIntegral u Dist a b = 1 - Dist x0 :=
  sorry -- Formal proof omitted, sketch in Appendix \ref{proof:lemma:integral_indicator_1d}
\end{lstlisting}

\textbf{Lean 4 Theorem \ref{thm:fsd_iff_1D} (FSD Equivalence (1D)):}
\begin{lstlisting}[caption={FSD Equivalence (1D)}, label={lst:1d_fsd_iff}]
theorem fsd_iff_integral_indicator_ge_1D {F G : \R \to \R} {a b : \R} (hab : a < b)
    (hFa : F a = 0) (hGa : G a = 0) (hFb : F b = 1) (hGb : G b = 1) :
    (\forall x \in Icc a b, F x \le G x) \iff
    (\forall x0 \in Ioo a b,
      let u := fun x => if x > x0 then (1 : \R) else 0
      -- The definition of u here matches the condition P in riemannStieltjesIntegral
      riemannStieltjesIntegral u F a b \ge riemannStieltjesIntegral u G a b) :=
  sorry -- Formal proof omitted, sketch in Appendix \ref{proof:thm:fsd_iff_1d}
\end{lstlisting}

\textbf{Lean 4 Definition \ref{def:Nd_Vector} (N-dimensional Vector of Reals):}
\begin{lstlisting}[caption={$N$-dimensional Vector of Reals ($\RVector$)}, label={lst:Vector_nd}]
def RVector (n : \N) := Fin n \to \R
\end{lstlisting}

\textbf{Lean 4 Definition \ref{def:vector relations} (Vector Relations):}
\begin{lstlisting}[caption={$N$-dimensional Vector Relations}, label={lst:Vector_relations_nd}]
-- Assuming 'n : \N' is a parameter for RVector n
def VecLT {n : \N} (x y : RVector n) : Prop := \forall i : Fin n, x i < y i
def VecLE {n : \N} (x y : RVector n) : Prop := \forall i : Fin n, x i \le y i
def allGt {n : \N} (x y : RVector n) : Prop := \forall i : Fin n, x i > y i

-- Instances for notation like x < y or x \le y can be defined
instance {n : \N} : LT (RVector n) := \<lt\>
instance {n : \N} : LE (RVector n) := \<le\>
\end{lstlisting}

\textbf{Lean 4 Definition \ref{def:N-dimensional Rectangles} ($N$-dimensional Rectangles):}
\begin{lstlisting}[caption={$N$-dimensional Rectangles ($\Iccn$, $\Ioon$)}, label={lst:Rectangles_nd}]
def closedRectangleND {n : \N} (a b : RVector n) : Set (RVector n) :=
  {x | \forall i, a i \le x i \wedge x i \le b i}

def Icc_n {n : \N} (a b : RVector n) : Set (RVector n) :=
  {x | RVector.le a x \wedge RVector.le x b}

def Ioo_n {n : \N} (a b : RVector n) : Set (RVector n) :=
  {x | RVector.lt a x \wedge RVector.lt x b}
\end{lstlisting}

\textbf{Lean 4 Definition \ref{def:Special Vector Constructions} (Special Vector Constructions):}
\begin{lstlisting}[caption={Special Vector Constructions ($\mixedVec$, $\text{replace}$, $\text{midpoint}$)}, label={lst:Special_Vector_Constructions_nd}]
def mixedVector {n : \N} (x0 b : RVector n) (s : Finset (Fin n)) : RVector n :=
  fun i => if i \in s then x0 i else b i

def replace_comp {n : \N} (x : RVector n) (j : Fin n) (val : \R) : RVector n :=
  fun i => if i = j then val else x i

noncomputable def midpoint {n : \N} (x y : RVector n) : RVector n :=
  fun i => (x i + y i) / 2
\end{lstlisting}

\textbf{Lean 4 Definition \ref{def:indicatorURO_ND} (Indicator Function for Upper-Right Orthant):}
\begin{lstlisting}[caption={Indicator Function for Upper-Right Orthant ($\indicatorURO$)}, label={lst:indicatorURO_nd}]
noncomputable def indicatorUpperRightOrthant {n : \N} (x0 : RVector n) (x : RVector n) : \R :=
  haveI : Decidable (allGt x x0) := Classical.propDecidable _
  if allGt x x0 then 1 else 0
\end{lstlisting}

\textbf{Lean 4 Definition \ref{def:survivalProbN_ND} ($N$-dimensional Survival Probability):}
\begin{lstlisting}[caption={$N$-dimensional Survival Probability ($\survivalProbN$)}, label={lst:survivalProbN_nd}]
noncomputable def survivalProbN {n : \N} (Dist : RVector n \to \R) (x0 b : RVector n) : \R :=
  1 - Finset.sum ((Finset.powerset (Finset.univ : Finset (Fin n))) \ {\empty})
      (fun s => (-1)^(s.card + 1) * Dist (mixedVector x0 b s))
\end{lstlisting}

\textbf{Lean 4 Definition \ref{def:riemannStieltjesND_ND} ($N$-dim Riemann-Stieltjes Integral for Orthant Indicators):}
\begin{lstlisting}[caption={$N$-dim Riemann-Stieltjes Integral for Orthant Indicators ($\riemannStieltjesND$)}, label={lst:riemannStieltjesND_nd}]
noncomputable def riemannStieltjesIntegralND {n : \N} (u : RVector n \to \R)
    (Dist : RVector n \to \R) (a b : RVector n) (x\0 : RVector n)
    (h_x\0 : x\0 \in Icc_n a b)
    (h_u : \forall x \in Ioo_n a b, u x = indicatorUpperRightOrthant x\0 x) : \R :=
  haveI : Decidable (\exists x\0' \in Icc_n a b, \forall x \in Ioo_n a b, u x =
indicatorUpperRightOrthant x\0' x) :=
    Classical.propDecidable _
  if h : \exists x\0' \in Icc_n a b, \forall x \in Ioo_n a b, u x = indicatorUpperRightOrthant x\0' x then
    -- Use x\0 directly, as it satisfies the condition by h_u
    survivalProbN Dist x\0 b
  else
    0\end{lstlisting}

\textbf{Lean 4 Lemma \ref{lem:Ioo_subset_Icc_self_n} (Subset Relation between Rectangles):}
\begin{lstlisting}[caption={Subset Relation between Rectangles ($\Ioon \subseteq \Iccn$)}, label={lst:lem_Ioo_subset_Icc_self_n}]
lemma Ioo_n_subset_Icc_n {n : \N} {a b : RVector n} : Ioo_n a b \sub Icc_n a b :=
  sorry -- Formal proof omitted, sketch in Appendix \ref{proof:lem:Ioo_subset_Icc_self_n}
\end{lstlisting}

\textbf{Lean 4 Lemma \ref{lem:uniqueness_ND} (Uniqueness of Indicator Function Parameter ($N$D)):}
\begin{lstlisting}[caption={Uniqueness of Indicator Function Parameter ($N$D)}, label={lst:lem_uniqueness_nd}]
lemma uniqueness_of_indicatorUpperRightOrthant_x\0_on_open {n : \N} {a b x\1 x\2 : RVector
n}
    (hab : \forall i, a i < b i)
    (hx\1_mem : x\1 \in Ioo_n a b) (hx\2_mem : x\2 \in Ioo_n a b)
    (h_eq_fn : \forall x \in Ioo_n a b, indicatorUpperRightOrthant x\1 x =
indicatorUpperRightOrthant x\2 x) :
    x\1 = x\2 :=
  sorry -- Formal proof omitted, sketch in Appendix \ref{proof:lem:uniqueness_nd}
\end{lstlisting}

\textbf{Lean 4 Lemma \ref{lem:integral_for_indicator_ND} (Integral for Indicator Function ($N$D)):}
\begin{lstlisting}[caption={Integral for Indicator Function ($N$D)}, label={lst:lem_integral_for_indicator_nd}]
lemma integral_for_indicatorUpperRightOrthant {n : \N} {a b : RVector n}
    (hab : \forall i, a i < b i)
    {u : RVector n \to \R} {Dist : RVector n \to \R} {x\0 : RVector n}
    (hx\0_mem : x\0 \in Ioo_n a b)
    (h_u_def : \forall x \in Icc_n a b, u x = indicatorUpperRightOrthant x\0 x) :
    riemannStieltjesIntegralND u Dist a b x\0 (Ioo_subset_Icc_self_n hx\0_mem)
      (fun x hx => h_u_def x (Ioo_subset_Icc_self_n hx)) = survivalProbN Dist x\0 b :=
  sorry -- Formal proof omitted, sketch in Appendix \ref{proof:lem:integral_for_indicator_nd}
\end{lstlisting}

\textbf{Lean 4 Theorem \ref{thm:fsd_equivalence_ND} (FSD Equivalence for Indicator Functions ($N$D)):}
\begin{lstlisting}[caption={FSD Equivalence for Indicator Functions ($N$D)}, label={lst:thm_fsd_equivalence_nd}]
theorem fsd_nd_iff_integral_indicatorUpperRightOrthant_ge {n : \N} {F G : RVector n \to \R}
    {a b : RVector n} (hab : \forall i, a i < b i) :
    (\forall x\0 \in Ioo_n a b, survivalProbN F x\0 b \ge survivalProbN G x\0 b) \lr
    (\forall (x\0 : RVector n) (hx\0 : x\0 \in Ioo_n a b),
      riemannStieltjesIntegralND (indicatorUpperRightOrthant x\0) F a b x\0
(Ioo_subset_Icc_self_n hx\0) (fun x _ => rfl) \ge
      riemannStieltjesIntegralND (indicatorUpperRightOrthant x\0) G a b x\0
(Ioo_subset_Icc_self_n hx\0) (fun x _ => rfl)) := by
sorry -- Formal proof omitted, sketch in Appendix \ref{proof:thm:fsd_equivalence_nd}
\end{lstlisting}

\subsection{Decidability and Classical Reasoning in Lean 4}\label{classical reasoning}
In definitions such as \lean{riemannStieltjesIntegral} (Definition \ref{def:Riemann-Stieltjes Integral_1D}) and \lean{indicatorUpperRightOrthant} (Definition \ref{def:indicatorURO_ND}), we employ classical reasoning techniques in Lean 4. This section explains the purpose and implications of these techniques.

\textbf{Decidability in Lean's Logic:}
Lean's underlying logic is constructive (intuitionistic). In constructive logic, to assert a proposition $P$, one must provide evidence (a proof) for $P$. Similarly, to use a proposition $P$ in a conditional statement like "if $P$ then $a$ else $b$," one needs to first establish that $P$ is decidable—meaning there exists an algorithm to determine whether $P$ is true or false.

\textbf{The Role of \lean{Classical.propDecidable}:}
The proposition $P$ in Definition \ref{def:Riemann-Stieltjes Integral_1D}:
\begin{lstlisting}
P := \exists x\0 \in Ioo a b, \forall x \in Icc a b, u x = if x > x\0 then 1 else 0
\end{lstlisting}
asserts the existence of a point $x_0$ that characterizes the function $u$ as a specific type of indicator function. For an arbitrary function $u: \R \to \R$, determining the truth of this existential statement algorithmically would require examining an uncountable set of potential $x_0$ values, which is not computationally feasible.

The declaration \lean{haveI : Decidable P := Classical.propDecidable P} invokes an axiom from classical logic: the law of excluded middle ($P \lor \neg P$). By assuming this axiom, \lean{Classical.propDecidable} provides a witness that $P$ is decidable, allowing us to use $P$ in a conditional expression without providing an explicit algorithm to decide $P$. This is a standard technique in Lean for handling propositions that are not constructively decidable.

\textbf{Noncomputability and \lean{Classical.choose}:}
When \lean{Classical.propDecidable} is used for a proposition like $P \equiv \exists y, Q(y)$, and if $P$ is true, we might need to obtain the witness $y$. The construct \lean{Classical.choose hP} (where \lean{hP} is a proof that $P$ holds) returns such a witness. Again, this relies on classical axioms, as there might not be a constructive way to extract the witness algorithmically. Functions that use \lean{Classical.choose} are labeled \lean{noncomputable} in Lean, indicating that they cannot be directly executed as algorithms.

\textbf{Theoretical Justification and Practical Implications:}
The use of classical reasoning is standard in most economic theory, including stochastic dominance. Our aim is to formalize these established mathematical concepts. Lean's framework allows us to do so rigorously by explicitly marking the points where non-constructive elements are introduced.
For our purposes, this is acceptable because:
\begin{enumerate}
    \item We are formalizing mathematical theory, where classical existence proofs are standard.
    \item The core FSD theorems relate properties of distribution functions; they are not primarily about computing specific integral values for arbitrary functions but about establishing equivalences between different characterizations of stochastic dominance.
    \item Even if some definitions are noncomputable, the resulting theorems can still be applied. For instance, if we can independently (and constructively) prove that a function $u$ *is* of the required indicator form, we can compute the expected value directly without relying on \lean{Classical.choose}.
\end{enumerate}
This approach allows us to leverage the power of classical mathematics within a formal system, ensuring logical rigor.

\textbf{References for Classical Logic in Lean}

For readers interested in deeper exploration of these concepts, the following references are recommended:

1. Avigad, J., de Moura, L., \& Kong, S. (2023). "Theorem Proving in Lean 4." \url{https://leanprover.github.io/theorem_proving_in\_lean4/} - See especially Chapter 6 on Propositions and Proofs, and Chapter 10 on Classical Logic.

2. de Moura, L., Ebner, G., Roesch, J.,\& Ullrich, S. (2021). "The Lean 4 Theorem Prover and Programming Language." In Automated Deduction – CADE 28 (pp. 625-635). Springer, Cham.

3. Carneiro, M. (2019). "Formalizing computability theory via partial recursive functions." In International Conference on Interactive Theorem Proving (pp. 12:1-12:17). Schloss Dagstuhl-Leibniz-Zentrum für Informatik.

4. Gonthier, G., Ziliani, B., Nanevski, A., \& Dreyer, D. (2013). "How to make ad hoc proof automation less ad hoc." Journal of Functional Programming, 23(4), 357-401. (For a broader perspective on proof automation in dependent type theory).

\subsection{Detailed Proof Sketches}
The following sketches outline the main arguments used in the formal Lean proofs for key lemmas and theorems. They use Lean-style tactic annotations (e.g., \tactic{intro}, \tactic{apply}, \tactic{linarith}) to highlight the reasoning steps.

\subsubsection{Proof Sketch of Lemma \ref{lemma:uniqueness_1D} (Uniqueness of Indicator Point (1D))}\label{proof:lemma:uniqueness_1d}
\textbf{Goal:} Given $a < b$, $x_1, x_2 \in \Ioo{a}{b}$, and $\forall x \in \Icc{a}{b}, (\text{if } x > x_1 \text{ then } 1 \text{ else } 0) = (\text{if } x > x_2 \text{ then } 1 \text{ else } 0)$, prove that $x_1 = x_2$.

\textbf{Proof by Contradiction:}
\begin{enumerate}
    \item We assume $x_1 \neq x_2$ and aim to derive a contradiction.

    \item \tactic{by\_contra h\_neq}: Introduce the hypothesis that $x_1 \neq x_2$.

    \item \tactic{have h\_lt\_or\_gt : $x_1 < x_2 \lor x_2 < x_1$ := Ne.lt\_or\_lt h\_neq}: Since $x_1 \neq x_2$, either $x_1 < x_2$ or $x_2 < x_1$.

    \item \tactic{rcases h\_lt\_or\_gt with h\_lt | h\_gt}: Split into two cases based on the ordering.

    \item \textbf{Case 1: $x_1 < x_2$ (with hypothesis \tactic{h\_lt})}
        \begin{enumerate}
            \item \tactic{let z := $(x_1 + x_2)/2$}: Define the midpoint between $x_1$ and $x_2$.

            \item \tactic{have hz\_mem\_Ioo : $z \in \Ioo a b$}: Prove that $z \in (a,b)$:
                \begin{enumerate}
                    \item \tactic{constructor}: Split into proving $a < z$ and $z < b$.

                    \item For $a < z$, we calculate:
                    \begin{align*}
                        a &< x_1 \quad \text{(by $\tactic{hx}_1\tactic{\_mem.1}$)} \\
                        &= \frac{x_1 + x_1}{2} \quad \text{(by \tactic{ring})} \\
                        &< \frac{x_1 + x_2}{2} \quad \text{(by \tactic{linarith} using \tactic{h\_lt})} \\
                        &= z \quad \text{(by definition)}
                    \end{align*}

                    \item For $z < b$, we calculate:
                    \begin{align*}
                        z &= \frac{x_1 + x_2}{2} \quad \text{(by definition)} \\
                        &< \frac{x_2 + x_2}{2} \quad \text{(by \tactic{linarith} using \tactic{h\_lt})} \\
                        &= x_2 \quad \text{(by \tactic{ring})} \\
                        &< b \quad \text{(by $\tactic{hx}_2\tactic{\_mem.2}$)}
                    \end{align*}
                \end{enumerate}

            \item \tactic{have hz\_mem\_Icc : $z \in \Icc a b$ := Ioo\_subset\_Icc\_self hz\_mem\_Ioo}: Since $z \in (a,b)$, we also have $z \in [a,b]$.

            \item \tactic{specialize h\_eq\_fn z hz\_mem\_Icc}: Apply our equality hypothesis to the point $z$.

            \item \tactic{have h\_z\_gt\_x$_1$ : $z > x_1$ := by unfold z; linarith [h\_lt]}: Prove $z > x_1$ using the definition of $z$ and $x_1 < x_2$.

            \item \tactic{have h\_z\_lt\_x$_2$ : $z < x_2$ := by unfold z; linarith [h\_lt]}: Prove $z < x_2$ similarly.

            \item \tactic{have h\_z\_not\_gt\_x$_2$ : $\lnot(z > x_2)$ := not\_lt.mpr (le\_of\_lt h\_z\_lt\_x$_2$)}: Since $z < x_2$, we have $\lnot(z > x_2)$.

            \item \tactic{simp only [if\_pos h\_z\_gt\_x$_1$, if\_neg h\_z\_not\_gt\_x$_2$] at h\_eq\_fn}: Simplify the equality using the facts above.
                \begin{itemize}
                    \item Since $z > x_1$, the left side evaluates to $1$.
                    \item Since $\lnot(z > x_2)$, the right side evaluates to $0$.
                    \item The equation becomes $1 = 0$.
                \end{itemize}

            \item \tactic{linarith [h\_eq\_fn]}: This gives a contradiction since $1 \neq 0$.
        \end{enumerate}

    \item \textbf{Case 2: $x_2 < x_1$ (with hypothesis \tactic{h\_gt})}
        \begin{enumerate}
            \item \tactic{let z := $(x_1 + x_2)/2$}: Define the midpoint between $x_1$ and $x_2$.

            \item \tactic{have hz\_mem\_Ioo : $z \in \Ioo a b$}: Prove that $z \in (a,b)$:
                \begin{enumerate}
                    \item For $a < z$, we calculate:
                    \begin{align*}
                        a &< x_2 \quad \text{(by $\tactic{hx}_2\tactic{\_mem.1}$)} \\
                        &= \frac{x_2 + x_2}{2} \quad \text{(by \tactic{ring})} \\
                        &< \frac{x_1 + x_2}{2} \quad \text{(by \tactic{linarith} using \tactic{h\_gt})} \\
                        &= z \quad \text{(by definition)}
                    \end{align*}

                    \item For $z < b$, we calculate:
                    \begin{align*}
                        z &= \frac{x_1 + x_2}{2} \quad \text{(by definition)} \\
                        &< \frac{x_1 + x_1}{2} \quad \text{(by \tactic{linarith} using \tactic{h\_gt})} \\
                        &= x_1 \quad \text{(by \tactic{ring})} \\
                        &< b \quad \text{(by $\tactic{hx}_1\tactic{\_mem.2}$)}
                    \end{align*}
                \end{enumerate}

            \item \tactic{have hz\_mem\_Icc : $z \in \Icc a b$ := Ioo\_subset\_Icc\_self hz\_mem\_Ioo}: Since $z \in (a,b)$, we also have $z \in [a,b]$.

            \item \tactic{specialize h\_eq\_fn z hz\_mem\_Icc}: Apply our equality hypothesis to the point $z$.

            \item \tactic{have h\_z\_lt\_x$_1$ : $z < x_1$ := by unfold z; linarith [h\_gt]}: Prove $z < x_1$ using the definition of $z$ and $x_2 < x_1$.

            \item \tactic{have h\_z\_gt\_x$_2$ : $z > x_2$ := by unfold z; linarith [h\_gt]}: Prove $z > x_2$ similarly.

            \item \tactic{have h\_z\_not\_gt\_x$_1$ : $\lnot(z > x_1)$ := not\_lt.mpr (le\_of\_lt h\_z\_lt\_x$_1$)}: Since $z < x_1$, we have $\lnot(z > x_1)$.

            \item \tactic{simp only [if\_neg h\_z\_not\_gt\_x$_1$, if\_pos h\_z\_gt\_x$_2$] at h\_eq\_fn}: Simplify the equality using the facts above.
                \begin{itemize}
                    \item Since $\lnot(z > x_1)$, the left side evaluates to $0$.
                    \item Since $z > x_2$, the right side evaluates to $1$.
                    \item The equation becomes $0 = 1$.
                \end{itemize}

            \item \tactic{linarith [h\_eq\_fn]}: This gives a contradiction since $0 \neq 1$.
        \end{enumerate}
\end{enumerate}
Since both cases lead to a contradiction, the initial assumption $x_1 \neq x_2$ must be false. Thus, $x_1 = x_2$. \qed

\subsubsection{Proof Sketch of Lemma \ref{lemma:integral_indicator_1D} (Integral Calculation for Indicator Functions (1D))}\label{proof:lemma:integral_indicator_1d}
\textbf{Goal:} Given $a < b$, $u(x) = \text{if } x > x_0 \text{ then } 1 \text{ else } 0$ for $x \in \Icc{a}{b}$ (with $x_0 \in \Ioo{a}{b}$), prove \\ $\lean{riemannStieltjesIntegral}(u, \text{Dist}, a, b) = 1 - \text{Dist}(x_0)$.

\begin{enumerate}
    \item \tactic{have h\_u\_is\_P : $\exists x_0' \in \Ioo a b$ $\forall$ $x \in \Icc a b$ $u(x) = \text{if}~x > x_0'~\text{then}~1~\text{else}~0$ := by use x$_0$}: We show that $u$ satisfies the property $P$ from our integral definition by using $x_0$ as the witness.

    \item \tactic{dsimp [riemannStieltjesIntegral]}: Unfold the definition of the Riemann-Stieltjes integral.

    \item \tactic{rw [dif\_pos h\_u\_is\_P]}: Since we proved that $u$ satisfies property $P$, we evaluate the \texttt{if} statement to its \tactic{then} branch, giving us $1 - \text{Dist}(\text{Classical.choose}~h\_u\_is\_P)$.

    \item Now we need to show that the chosen $x_0'$ equals our given $x_0$:
        \begin{enumerate}
            \item \tactic{have h\_x$_0$\_eq : Classical.choose h\_u\_is\_P = x$_0$ := by}:
                \begin{enumerate}
                    \item \tactic{let x$_0'$ := Classical.choose h\_u\_is\_P}: Name the chosen value.

                    \item \tactic{have h\_spec' : $x_0' \in \Ioo a b \land \forall x \in \Icc a b$ $u(x) = \text{if}~x > x_0'~\text{then}~1~\text{else}~0$ := \\
                        Classical.choose\_spec h\_u\_is\_P}: Extract the properties of the chosen value.

                    \item \tactic{have h\_eq\_fn : $\forall x \in \Icc\,a\,b, (\text{if}~x > x_0'~\text{then}~1~\text{else}~0) = (\text{if}~x > x_0~\text{then}~1~\text{else}~0)$}: We prove that the indicator functions defined by $x_0'$ and $x_0$ are equal on $[a,b]$.

                        \item For any $x \in [a,b]$, we show:
                        \begin{align*}
                            (\text{if}~x > x_0'~\text{then}~1~\text{else}~0) &= u(x) \quad \text{(by \tactic{$h\_\text{spec}'.2$})} \\
                            &= (\text{if}~x > x_0~\text{then}~1~\text{else}~0) \quad \text{(by \tactic{$h\_u\_\text{def}$})}
                        \end{align*}

                    \item \tactic{exact uniqueness\_of\_indicator\_x$_0$ hab h\_spec'.1 hx$_0$\_mem h\_eq\_fn}: Apply Lemma \ref{lemma:uniqueness} to conclude that $x_0' = x_0$.
                \end{enumerate}
        \end{enumerate}

    \item \tactic{rw [h\_x$_0$\_eq]}: Substitute the equality $\text{Classical.choose}~h\_u\_is\_P = x_0$ into our goal, which transforms it to $1 - \text{Dist}(x_0) = 1 - \text{Dist}(x_0)$, which is true by reflexivity.
\end{enumerate}\qed
\subsubsection{Proof Sketch of Theorem \ref{thm:fsd_iff_1D} (FSD Equivalence (1D))}\label{proof:thm:fsd_iff_1d}
\textbf{Goal:} Given $F(a)=G(a)=0, F(b)=G(b)=1$, prove $(\forall x \in \Icc{a}{b}, F(x) \leq G(x)) \iff (\forall x_0 \in \Ioo{a}{b}, \text{let } u_0(x) = \indicator{(x_0,\infty)}(x), \lean{riemannStieltjesIntegral}(u_0, F, a, b) \geq \lean{riemannStieltjesIntegral}(u_0, G, a, b))$.

We prove both directions of the equivalence separately.

\begin{enumerate}
    \item \tactic{constructor}: Split the bidirectional implication into two parts.

    \item \textbf{Forward Direction ($\Rightarrow$):} Assume $F(x) \leq G(x)$ for all $x \in [a,b]$, prove the integral inequality.
        \begin{enumerate}
            \item \tactic{intro h\_dominance}: Introduce the hypothesis that $F(x) \leq G(x)$ for all $x \in [a,b]$.

            \item \tactic{intro x$_0$ hx$_0$\_mem}: Consider an arbitrary $x_0 \in (a,b)$.

            \item \tactic{let u := fun x => if x > x$_0$ then (1 : $\mathbb{R}$) else 0}: Define the indicator function $u$ for this $x_0$.

            \item \tactic{have h\_u\_def : $\forall x \in \Icc\,a\,b, u(x) = \text{if}~x > x_0~\text{then}~1~\text{else}~0$ := by intro x \_hx; rfl}: Confirm the definition of $u$.

            \item \tactic{have calc\_int\_F : riemannStieltjesIntegral u F a b = 1 - F x$_0$ := by \\
            apply integral\_for\_indicator hab hx$_0$\_mem h\_u\_def}: Calculate the integral with respect to $F$ using Lemma \ref{lemma:integral_indicator_1D}.

            \item \tactic{have calc\_int\_G : riemannStieltjesIntegral u G a b = 1 - G x$_0$ := by \\
            apply integral\_for\_indicator hab hx$_0$\_mem h\_u\_def}: Similarly for $G$.

            \item \tactic{dsimp only}: Unfold the \tactic{let} binding in the goal.

            \item \tactic{rw [calc\_int\_F, calc\_int\_G]}: Substitute the calculated integral values.

            \item \tactic{rw [ge\_iff\_le, sub\_le\_sub\_iff\_left]}: Simplify the inequality $1 - F(x_0) \geq 1 - G(x_0)$ to $F(x_0) \leq G(x_0)$.

            \item \tactic{apply h\_dominance}: Apply our main hypothesis. We need to show $x_0 \in [a,b]$.

            \item \tactic{exact Ioo\_subset\_Icc\_self hx$_0$\_mem}: Since $x_0 \in (a,b)$, we have $x_0 \in [a,b]$.
        \end{enumerate}

    \item \textbf{Backward Direction ($\Leftarrow$):} Assume the integral inequality for all indicators, prove $F(x) \leq G(x)$ for all $x \in [a,b]$.
        \begin{enumerate}
            \item \tactic{intro h\_integral\_indicator}: Introduce the hypothesis that for all $x_0 \in (a,b)$, the integral inequality holds.

            \item \tactic{intro x$_0$ hx$_0$\_mem\_Icc}: Consider an arbitrary $x_0 \in [a,b]$.

            \item \tactic{by\_cases h\_eq\_a : x$_0$ = a}: Handle the case where $x_0 = a$.
                \begin{itemize}
                    \item \tactic{rw [h\_eq\_a, hFa, hGa]}: If $x_0 = a$, then $F(x_0) = F(a) = 0$ and $G(x_0) = G(a) = 0$, so $F(x_0) \leq G(x_0)$ becomes $0 \leq 0$, which is true.
                \end{itemize}

            \item \tactic{by\_cases h\_eq\_b : x$_0$ = b}: Handle the case where $x_0 = b$.
                \begin{itemize}
                    \item \tactic{rw [h\_eq\_b, hFb, hGb]}: If $x_0 = b$, then $F(x_0) = F(b) = 1$ and $G(x_0) = G(b) = 1$, so $F(x_0) \leq G(x_0)$ becomes $1 \leq 1$, which is true.
                \end{itemize}

            \item Now we handle the case where $x_0 \in (a,b)$:
                \begin{enumerate}
                    \item \tactic{push\_neg at h\_eq\_a h\_eq\_b}: Transform $\lnot(x_0 = a)$ to $x_0 \neq a$ and $\lnot(x_0 = b)$ to $x_0 \neq b$.

                    \item \tactic{have hx$_0$\_mem\_Ioo : $x_0 \in \Ioo\,a\,b$ := $\langle$lt\_of\_le\_of\_ne hx$_0$\_mem\_Icc.1 (Ne.symm h\_eq\_a), lt\_of\_le\_of\_ne hx$_0$\_mem\_Icc.2 h\_eq\_b$\rangle$}: Since $a \leq x_0 \leq b$ and $x_0 \neq a$ and $x_0 \neq b$, we have $a < x_0 < b$, i.e., $x_0 \in (a,b)$.

                    \item \tactic{specialize h\_integral\_indicator x$_0$ hx$_0$\_mem\_Ioo}: Apply our hypothesis to $x_0 \in (a,b)$.

                    \item \tactic{let u := fun x => if x > x$_0$ then (1 : $\mathbb{R}$) else 0}: Define the indicator function for $x_0$.

                    \item \tactic{have h\_u\_def : $\forall x \in \Icc\,a\,b, u(x) = \text{if}~x > x_0~\text{then}~1~\text{else}~0$ := by intro x \_hx; rfl}: Confirm the definition.

                    \item \tactic{have calc\_int\_F : riemannStieltjesIntegral u F a b = 1 - F x$_0$ := by \\
                    apply integral\_for\_indicator hab hx$_0$\_mem\_Ioo h\_u\_def}: Calculate the integral for $F$.

                    \item \tactic{have calc\_int\_G : riemannStieltjesIntegral u G a b = 1 - G x$_0$ := by \\
                    apply integral\_for\_indicator hab hx$_0$\_mem\_Ioo h\_u\_def}: Calculate the integral for $G$.

                    \item \tactic{dsimp only at h\_integral\_indicator}: Unfold the \texttt{let} binding in the hypothesis.

                    \item \tactic{rw [calc\_int\_F, calc\_int\_G] at h\_integral\_indicator}: Substitute the calculated integral values.

                    \item \tactic{rw [ge\_iff\_le, sub\_le\_sub\_iff\_left] at h\_integral\_indicator}: Simplify the inequality in the hypothesis to $F(x_0) \leq G(x_0)$.

                    \item \tactic{exact h\_integral\_indicator}: The hypothesis now exactly matches our goal.
                \end{enumerate}
        \end{enumerate}
\end{enumerate}

This completes the proof of the equivalence. We've shown that $F(x) \leq G(x)$ for all $x \in [a,b]$ if and only if the integral inequality holds for all indicator functions. \qed

\subsubsection{Proof Sketch of Lemma \ref{lem:Ioo_subset_Icc_self_n} ($\Ioon(a,b) \subseteq \Iccn(a,b)$)}\label{proof:lem:Ioo_subset_Icc_self_n}
\textbf{Goal:} For $a, b \in \RVector(n)$, prove $\Ioon(a,b) \subseteq \Iccn(a,b)$.

\begin{enumerate}
    \item \tactic{intro x hx}: Introduces an arbitrary element $x$ and a hypothesis \Lean{hx} stating $x \in \Ioon(a,b)$. By definition of $\Ioon(a,b)$, \tactic{hx} means $a < x \land x < b$. This can be accessed as \tactic{hx.1} ($a < x$) and \tactic{hx.2} ($x < b$).
    \item \tactic{constructor}: The goal is to prove $x \in \Iccn(a,b)$, which by definition is $a \leq x \land x \leq b$. The \tactic{constructor} tactic splits this conjunction into two subgoals:
    \begin{enumerate}
        \item $a \leq x$
        \item $x \leq b$
    \end{enumerate}
    \item For the first subgoal ($a \leq x$, which means $\forall i, a(i) \leq x(i)$):
    \begin{enumerate}
        \item \tactic{intro i}: Introduces an arbitrary index $i \in \Fin(n)$. The goal becomes $a(i) \leq x(i)$.
        \item \tactic{exact le\_of\_lt (hx.1 i)}: From \tactic{hx.1}, we have $a < x$, which means $\forall k, a(k) < x(k)$. So, specifically for index $i$, we have $a(i) < x(i)$. The lemma \tactic{le\_of\_lt} states that if $u < v$, then $u \leq v$. Thus, $a(i) < x(i)$ implies $a(i) \leq x(i)$. This completes the first subgoal.
    \end{enumerate}
    \item For the second subgoal ($x \leq b$, which means $\forall i, x(i) \leq b(i)$):
    \begin{enumerate}
        \item \tactic{intro i}: Introduces an arbitrary index $i \in \Fin(n)$. The goal becomes $x(i) \leq b(i)$.
        \item \tactic{exact le\_of\_lt (hx.2 i)}: From \tactic{hx.2}, we have $x < b$, which means $\forall k, x(k) < b(k)$. So, for index $i$, $x(i) < b(i)$. Using \tactic{le\_of\_lt}, this implies $x(i) \leq b(i)$. This completes the second subgoal.
    \end{enumerate}
\end{enumerate}
All goals are proven, so the lemma holds. \qed

\subsubsection{Proof Sketch of Lemma \ref{lem:uniqueness_ND} (Uniqueness of Indicator Function Parameter ($N$D))}\label{proof:lem:uniqueness_nd}
\textbf{Goal:} Given $a, b, x_1, x_2 \in \RVector(n)$ with $\forall i, a(i) < b(i)$ ($h_{ab\_open}$), $x_1, x_2 \in \Ioon(a,b)$, and $\forall x \in \Ioon(a,b), \indicatorURO(x_1, x) = \indicatorURO(x_2, x)$, prove $x_1 = x_2$.

The proof is by contradiction.

\begin{enumerate}
    \item \tactic{by\_contra h\_neq}: Assume $x_1 \neq x_2$ for contradiction. \tactic{h\_neq} is this hypothesis.
    \item \tactic{have h\_exists\_diff : $\exists$ j, x1 j $\neq$ x2 j := by}: This block proves that if $x_1 \neq x_2$, then there must be an index $j$ where their components differ.
    \begin{enumerate}
        \item \tactic{by\_contra h\_all\_eq}: Inner proof by contradiction. Assume $\neg (\exists j, x_1(j) \neq x_2(j))$, which is equivalent to $\forall j, x_1(j) = x_2(j)$. This is \tactic{h\_all\_eq}.
        \item \tactic{push\_neg at h\_all\_eq}: Transforms \tactic{h\_all\_eq} from $\neg (\exists j, x_1(j) \neq x_2(j))$ to $\forall j, \neg (x_1(j) \neq x_2(j))$, which simplifies to $\forall j, x_1(j) = x_2(j)$.
        \item \tactic{have h\_eq : x1 = x2 := by funext i; exact h\_all\_eq i}: If all components are equal ($\forall i, x_1(i) = x_2(i)$), then by function extensionality (\tactic{funext i}), the vectors $x_1$ and $x_2$ are equal. This is \tactic{h\_eq}.
        \item \tactic{exact h\_neq h\_eq}: This leads to a contradiction. We have \Lean{h\_neq}: $x_1 \neq x_2$ and \tactic{h\_eq}: $x_1 = x_2$.
    \end{enumerate}
    \item \tactic{rcases h\_exists\_diff with ⟨j, h\_diff\_at\_j⟩}: Destructures the existential hypothesis \tactic{h\_exists\_diff}. This gives an index $j$ and a hypothesis \tactic{h\_diff\_at\_j}: $x_1(j) \neq x_2(j)$.
    \item \tactic{have h\_lt\_or\_gt : x1 j < x2 j $\vee$ x2 j < x1 j := by exact Ne.lt\_or\_lt h\_diff\_at\_j}: Since $x_1(j) \neq x_2(j)$, by trichotomy for real numbers, either $x_1(j) < x_2(j)$ or $x_2(j) < x_1(j)$. This is captured by \tactic{Ne.lt\_or\_lt}.
    \item \tactic{rcases h\_lt\_or\_gt with h\_lt | h\_gt}: Splits the proof into two cases based on the disjunction \tactic{h\_lt\_or\_gt}.
    \begin{enumerate}
        \item \textbf{Case 1: \tactic{h\_lt : x1 j < x2 j}}
        \begin{enumerate}
            \item \tactic{let z1 := (x1 j + x2 j) / 2}: Define $z_1$ as the midpoint of $x_1(j)$ and $x_2(j)$.
            \item \tactic{have hz1\_gt\_x1j : z1 > x1 j := by ...}: Prove $z_1 > x_1(j)$.
            \begin{enumerate}
                \item \tactic{dsimp [z1]}: Unfold definition of $z_1$.
                \item \tactic{have h\_two\_pos : (0 : R) < 2 := by norm\_num}: Establish $2 > 0$. \Lean{norm\_num} simplifies numerical expressions.
                \item \tactic{rw [gt\_iff\_lt, lt\_div\_iff0 h\_two\_pos]}: Rewrite $z_1 > x_1(j)$ to $x_1(j) < z_1$, then to $x_1(j) \cdot 2 < x_1(j) + x_2(j)$ using $2 > 0$.
                \item \tactic{rw [mul\_comm, two\_mul]}: Rewrite $x_1(j) \cdot 2$ to $x_1(j) + x_1(j)$.
                \item \tactic{rw [Real.add\_lt\_add\_iff\_left (x1 j)]}: Cancel $x_1(j)$ from both sides of $x_1(j) + x_1(j) < x_1(j) + x_2(j)$. Goal becomes $x_1(j) < x_2(j)$.
                \item \tactic{exact h\_lt}: This is exactly the hypothesis for this case.
            \end{enumerate}
            \item \tactic{have hz1\_lt\_x2j : z1 < x2 j := by ...}: Similarly, prove $z_1 < x_2(j)$.
            \item \tactic{let z : RVector n := fun i => if i = j then z1 else (max (x1 i) (x2 i) + b i) / 2}: Construct the test vector $z$. For component $j$, $z(j)=z_1$. For $i \neq j$, $z(i)$ is the midpoint of $\max(x_1(i), x_2(i))$ and $b(i)$. This ensures $z(i)$ is greater than $x_1(i)$ and $x_2(i)$, and less than $b(i)$.
            \item \tactic{have hz\_mem\_Ioo : z $\in$ Ioo\_n a b := by ...}: Prove $z \in \Ioon(a,b)$. This involves showing $a(i) < z(i)$ and $z(i) < b(i)$ for all $i$, considering $i=j$ and $i \neq j$ separately using \tactic{by\_cases}, \tactic{calc} for inequalities, and properties like \tactic{le\_max\_left}, \tactic{max\_lt}.
            \item \tactic{specialize h\_eq\_fn z hz\_mem\_Ioo}: Apply the main hypothesis $\forall x \in \Ioon(a,b), \\ \indicatorURO(x_1, x) = \indicatorURO(x_2, x)$ to our specific $z$. Now \tactic{h\_eq\_fn} is $\indicatorURO(x_1, z) = \indicatorURO(x_2, z)$.
            \item \tactic{have ind1 : indicatorURO x1 z = 1 := by ...}: Prove $\indicatorURO(x_1, z) = 1$.
            \begin{enumerate}
                \item \tactic{unfold indicatorURO}: Expand definition.
                \item \tactic{apply if\_pos}: We need to show $\allGt(z, x_1)$ (i.e., $\forall i', z(i') > x_1(i')$).
                \item \tactic{intro i'}: Take arbitrary $i'$.
                \item \Lean{by\_cases h\_eq\_j' : i' = j}: Case on $i'=j$. If $i'=j$, $z(j)=z_1 > x_1(j)$ by \tactic{hz1\_gt\_x1j}. If $i' \neq j$, $z(i') = (\max(x_1(i'), x_2(i')) + b(i'))/2$. This is greater than $x_1(i')$ because $\max(x_1(i'), x_2(i')) \geq x_1(i')$ and $b(i') > x_1(i')$ (since $x_1 \in \Ioon(a,b)$). Proof uses \tactic{add\_lt\_add\_of\_le\_of\_lt}.
            \end{enumerate}
            \item \tactic{have ind2 : indicatorURO x2 z = 0 := by ...}: Prove $\indicatorURO(x_2, z) = 0$.
            \begin{enumerate}
                \item \tactic{unfold indicatorURO}: Expand definition.
                \item \tactic{apply if\_neg}: We need to show $\neg \allGt(z, x_2)$. This means $\exists i', \neg (z(i') > x_2(i'))$.
                \item \tactic{intro h\_all\_gt\_z\_x2}: Assume $\allGt(z, x_2)$ for contradiction.
                \item \tactic{specialize h\_all\_gt\_z\_x2 j}: This gives $z(j) > x_2(j)$.
                \item \tactic{simp only [z, if\_true] at h\_all\_gt\_z\_x2}: Since $z(j)=z_1$, this becomes $z_1 > x_2(j)$.
                \item \tactic{linarith [hz1\_lt\_x2j, h\_all\_gt\_z\_x2]}: We have $hz1\_lt\_x2j: z_1 < x_2(j)$ and $h\_all\_gt\_z\_x2: z_1 > x_2(j)$. This is a contradiction. \Lean{linarith} resolves it.
            \end{enumerate}
            \item \tactic{rw [ind1, ind2] at h\_eq\_fn}: Substitute results into \tactic{h\_eq\_fn}. It becomes $1=0$.
            \item \tactic{exact absurd h\_eq\_fn (by norm\_num)}: $1=0$ is absurd. \tactic{norm\_num} proves $1 \neq 0$.
        \end{enumerate}
        \item \textbf{Case 2: \tactic{h\_gt : x2 j < x1 j}}
        \begin{enumerate}
            \item The logic is symmetric to Case 1, swapping roles of $x_1$ and $x_2$. This time it will be shown that $\indicatorURO(x_1, z) = 0$ and $\indicatorURO(x_2, z) = 1$, leading to $0=1$, which is also absurd.
        \end{enumerate}
    \end{enumerate}
\end{enumerate}
Since both cases lead to a contradiction, the initial assumption \tactic{h\_neq} ($x_1 \neq x_2$) must be false. Thus $x_1 = x_2$.
 \qed

\subsubsection{Proof Sketch of Lemma \ref{lem:integral_for_indicator_ND} (Integral for Indicator Function ($N$D))}\label{proof:lem:integral_for_indicator_nd}
\textbf{Goal:} Given the specified hypotheses about $u$ being an orthant indicator function defined by $x_0 \in \Ioon(a,b)$, prove that
$\riemannStieltjesND(u, \text{Dist}, a, b, x_0, h_{x_0}', h_u') = \survivalProbN(\text{Dist}, x_0, b)$.

The goal is to show that the integral definition simplifies to $\survivalProbN(\text{Dist}, x_0, b)$.
\begin{enumerate}
    \item \tactic{have hx0\_mem\_Icc : x0 $\in$ Icc\_n a b := Ioo\_subset\_Icc\_self\_n hx0\_mem}: Establishes that $x_0 \in \Iccn(a,b)$ using Lemma \ref{lem:Ioo_subset_Icc_self_n}, since \tactic{hx0\_mem} states $x_0 \in \Ioon(a,b)$. This serves as \tactic{$h_{x_0}'$}.
    \item \tactic{have h\_match : $\exists$ x0' $\in$ Icc\_n a b, $\forall$ x $\in$ Ioo\_n a b, u x = indicatorURO x0' x := by ...}: This proves the condition in the \tactic{if} statement of the $\riemannStieltjesND$ definition.
    \begin{enumerate}
        \item \tactic{use x0, hx0\_mem\_Icc}: We claim $x_0$ is the $x_0'$ that satisfies the condition. We provide $x_0$ and the proof \tactic{hx0\_mem\_Icc} that $x_0 \in \Iccn(a,b)$.
        \item \tactic{intro x hx\_open}: We need to show $\forall x \in \Ioon(a,b), u(x) = \indicatorURO(x_0, x)$. So, take an arbitrary $x$ and assume $x \in \Ioon(a,b)$ (\tactic{hx\_open}).
        \item \tactic{have hx\_closed : x $\in$ Icc\_n a b := Ioo\_subset\_Icc\_self\_n hx\_open}: Show that this $x$ is also in $\Iccn(a,b)$ using Lemma \ref{lem:Ioo_subset_Icc_self_n}.
        \item \tactic{exact h\_u\_def x hx\_closed}: The main hypothesis \tactic{h\_u\_def} states $\forall y  \in \Iccn(a,b),\\ u(y) = \indicatorURO(x_0, y)$. Since $x \in \Iccn(a,b)$ (by \tactic{hx\_closed}), we can apply \tactic{h\_u\_def} to $x$, yielding \\ $u(x) = \indicatorURO(x_0, x)$, which is the goal. This part also serves as \tactic{$h_u'$}.
    \end{enumerate}
    \item \tactic{unfold riemannStieltjesIntegralND}: Expand the definition of $\riemannStieltjesND$. It is an \tactic{if} statement.
    \item \tactic{simp only [h\_match, dif\_pos]}:
    \begin{enumerate}
        \item \tactic{h\_match} is the proof that the condition of the \tactic{if} statement is true.
        \item \tactic{dif\_pos} is a lemma used to simplify an \tactic{if h : c then t else e} to \tactic{t} when \tactic{h : c} (i.e., $c$ is true).
        \item So, the expression $\riemannStieltjesND \dots$ simplifies to its first branch, which is $\survivalProbN(\text{Dist}, x_0, b)$. This matches the goal.
    \end{enumerate}
\end{enumerate}
 \qed

\subsubsection{Proof Sketch of Theorem \ref{thm:fsd_equivalence_ND} (FSD Equivalence for Indicator Functions ($N$D))}\label{proof:thm:fsd_equivalence_nd}
\textbf{Goal:} Prove the equivalence between the survival probability characterization of FSD and the expected utility characterization using orthant indicator utility functions.

The proof of equivalence is split into two directions.
\begin{enumerate}
    \item \tactic{constructor}: This tactic tells Lean to prove both implications of the $\iff$ statement.
    \item \textbf{Forward Direction ($\Rightarrow$):}
    Assume $\forall x_0 \in \Ioon(a,b), \survivalProbN(F, x_0, b) \geq \survivalProbN(G, x_0, b)$ (hypothesis \tactic{h\_survival\_dominance}).
    We need to show $\forall x_0 \in \Ioon(a,b),\\ \riemannStieltjesND(\indicatorURO(x_0), F, \dots)\\ \geq \riemannStieltjesND(\indicatorURO(x_0), G, \dots)$.
    \begin{enumerate}
        \item \tactic{intro x0 hx0\_mem}: Introduce an arbitrary $x_0$ and the hypothesis $x_0 \in \Ioon(a,b)$ (\tactic{hx0\_mem}).
        \item \tactic{let u := indicatorURO x0}: Define $u$ to be $\indicatorURO(x_0)$.
        \item \tactic{$\forall x \in Icc\_n a b, u x = indicatorURO x0 x := fun x_ \Rightarrow rfl$}: This states that $u$ is indeed the indicator function for $x_0$. \tactic{rfl} (reflexivity) suffices because $u$ is defined as $\indicatorURO(x_0)$. The arguments \tactic{x} and \tactic{\_} (an unused hypothesis that $x \in \Iccn(a,b)$) are for the universal quantifier.
        \item \tactic{have calc\_int\_F : riemannStieltjesND u F a b x0 ... = survivalProbN F x0 b := by apply integral\_for\_indicatorUpperRightOrthant hab hx0\_mem h\_u\_def}: By Lemma \ref{lem:integral_for_indicator_ND}, the integral for $F$ simplifies to the survival probability for $F$. The arguments \tactic{hab}, \tactic{hx0\_mem}, and \tactic{h\_u\_def} satisfy the premises of the lemma.
        \item \tactic{have calc\_int\_G : riemannStieltjesND u G a b x0 ... = survivalProbN G x0 b := \\by apply integral\_for\_indicatorUpperRightOrthant hab hx0\_mem h\_u\_def}: Similarly for $G$.
        \item \tactic{rw [calc\_int\_F, calc\_int\_G]}: Rewrite the goal using these two established equalities. The goal becomes $\survivalProbN(F, x_0, b) \geq \survivalProbN(G, x_0, b)$.
        \item \tactic{apply h\_survival\_dominance x0 hx0\_mem}: This is exactly the assumption \tactic{h\_survival\_dominance} applied to the current $x_0$ and \tactic{hx0\_mem}.
    \end{enumerate}
    \item \textbf{Backward Direction ($\Leftarrow$):}
    Assume $\forall x_0 \in \Ioon(a,b), \\\riemannStieltjesND(\indicatorURO(x_0), F, \dots)\\ \geq \riemannStieltjesND(\indicatorURO(x_0), G, \dots)$ (hypothesis \tactic{h\_integral\_indicator}).
    We need to show $\forall x_0 \in \Ioon(a,b), \survivalProbN(F, x_0, b) \geq \survivalProbN(G, x_0, b)$.
    \begin{enumerate}
        \item \tactic{intro x0 hx0\_mem}: Introduce an arbitrary $x_0$ and $x_0 \in \Ioon(a,b)$ (\tactic{hx0\_mem}).
        \item \tactic{specialize h\_integral\_indicator x0 hx0\_mem}: Apply the hypothesis \tactic{h\_integral\_indicator} to this specific $x_0$ and \tactic{hx0\_mem}. Now \tactic{h\_integral\_indicator} states\\ $\riemannStieltjesND(\indicatorURO(x_0), F, \dots)\\ \geq \riemannStieltjesND(\indicatorURO(x_0), G, \dots)$ for this particular $x_0$.
        \item \tactic{let u := indicatorURO x0}: Define $u$.
        \item \tactic{have h\_u\_def : $\forall x \in Icc\_n a b, u x = indicatorURO x0 x := fun x _ \Rightarrow rfl$}: As before.
        \item \tactic{have calc\_int\_F : riemannStieltjesND u F a b x0 ... = survivalProbN F x0 b := by apply integral\_for\_indicatorUpperRightOrthant hab hx0\_mem h\_u\_def}: As before.
        \item \tactic{have calc\_int\_G : riemannStieltjesND u G a b x0 ... = survivalProbN G x0 b := by apply integral\_for\_indicatorUpperRightOrthant hab hx0\_mem h\_u\_def}: As before.
        \item \tactic{rw [calc\_int\_F, calc\_int\_G] at h\_integral\_indicator}: Rewrite the terms in the specialized hypothesis \tactic{h\_integral\_indicator} using these equalities.\\ The hypothesis becomes $\survivalProbN(F, x_0, b) \geq \survivalProbN(G, x_0, b)$.
        \item \tactic{exact h\_integral\_indicator}: This transformed hypothesis is exactly what we need to prove for this direction.
    \end{enumerate}
\end{enumerate}
Both directions are proven, so the equivalence holds.
 \qed

\subsubsection{Implementation Challenges and Solutions}
From our experience implementing this geometric framework in Lean 4, several practical challenges arose:

\begin{itemize}
       \item \textbf{Challenge:} Managing proof complexity when working with existential propositions about functions matching the orthant indicator form.

    \textbf{Solution:} Structured the specialized integral definitions to accept explicit proof arguments (the hypotheses $h_{x_0}$ and $h_u$), avoiding repeated complex existence proofs.

    \item \textbf{Challenge:} Balancing the need for classical logic with the desire for extractable computational content.

    \textbf{Solution:} Isolated classical reasoning to specific components, leaving core algorithms that check empirical FSD amenable to extraction.
\end{itemize}

These implementation insights highlight how the geometric approach not only simplifies the mathematical theory but also leads to more practical and maintainable formal verification code in Lean 4.

\section{Acknowledgments}
I would like to thank the Mathlib community for their continuous development of mathematical libraries in Lean, which made this work possible.

\end{document}